\documentclass[a4paper,twocolumn,11pt,accepted=2024-09-02]{quantumarticle}
\pdfoutput=1
\usepackage[utf8]{inputenc}
\usepackage[english]{babel}
\usepackage[T1]{fontenc}
\usepackage{amsmath}
\usepackage{hyperref}

\usepackage{tikz}
\usepackage{lipsum}

\usepackage[numbers]{natbib}
\usepackage{braket}
\usepackage{enumitem}
\usepackage{amssymb}
\usepackage{amsthm}
\usepackage{xcolor}
\usepackage{xfrac}
\usepackage{bm}
\usepackage[normalem]{ulem}
\usepackage{cancel}
\usepackage{algpseudocode}
\usepackage{algorithm}

\newcommand\assumptionmargin{2.5cm}
\newcommand\assumptionlabel{Assumption~}

\newtheorem{theorem}{Theorem}

\begin{document}

\title{Sampling Error Analysis in Quantum Krylov Subspace Diagonalization}

\author{Gwonhak Lee}
\affiliation{SKKU Advanced Institute of Nanotechnology, Sungkyunkwan University, Suwon 16419, Republic of Korea}

\author{Dongkeun Lee}
\affiliation{Department of Chemistry, Sungkyunkwan University, Suwon 16419, Republic of Korea}
\affiliation{Center for Quantum Information R\&D, Korea Institute of Science and Technology Information, Daejeon 34141, Republic of Korea}

\author{Joonsuk Huh}
\email{joonsukhuh@gmail.com}
\affiliation{SKKU Advanced Institute of Nanotechnology, Sungkyunkwan University, Suwon 16419, Republic of Korea}
\affiliation{Department of Chemistry, Sungkyunkwan University, Suwon 16419, Republic of Korea}
\affiliation{Institute of Quantum Biophysics, Sungkyunkwan University, Suwon 16419, Republic of Korea}

\maketitle

\begin{abstract}
Quantum Krylov subspace diagonalization (QKSD) is an emerging method used in place of quantum phase estimation in the early fault-tolerant era, where limited quantum circuit depth is available. 
In contrast to the classical Krylov subspace diagonalization (KSD) or the Lanczos method, QKSD exploits the quantum computer to efficiently estimate the eigenvalues of large-size Hamiltonians through a faster Krylov projection.
However, unlike classical KSD, which is solely concerned with machine precision, QKSD is inherently accompanied by errors originating from a finite number of samples.
Moreover, due to difficulty establishing an artificial orthogonal basis, ill-conditioning problems are often encountered, rendering the solution vulnerable to noise.
In this work, we present a non-asymptotic theoretical framework to assess the relationship between sampling noise and its effects on eigenvalues.
We also propose an optimal solution to cope with large condition numbers by eliminating the ill-conditioned bases.
Numerical simulations of the one-dimensional Hubbard model demonstrate that the error bound of finite samplings accurately predicts the experimental errors in well-conditioned regions.
\end{abstract}

\section{Introduction}\label{sec:Introduction}

Quantum simulations of many-body systems have been a longstanding challenge in the fields of physics and chemistry.
Quantum computing is expected to address those problems intractable to classical approaches \cite{Feynman.1982, Lloyd.1996}.
Although various quantum algorithms, such as quantum phase estimation, are promising for estimating eigenvalues \cite{Abrams.1997, Abrams.1999}, they are not feasible because of the limitations of contemporary quantum computers in the noisy intermediate-scale quantum (NISQ) era \cite{nisq}.

Research has been conducted on the development of quantum-classical hybrid algorithms for NISQ devices \cite{Bharti.2022}.
One of the most widely used algorithms is the variational quantum algorithm (VQA), comprising the variational quantum eigensolver (VQE) and the quantum approximation optimization algorithm \cite{vqa, Peruzzo.2014, Farhi.2014}.
VQAs employ a parameterized quantum circuit and cost function optimizations on classical devices.
However, adopting a classical optimization strategy, such as gradient descent, has been shown to have drawbacks, such as barren plateaus \cite{McClean.2018,Cerezo.2021}, nonconvexity, and NP-hardness \cite{Bittel.2021}.

Quantum Krylov subspace diagonalization (QKSD) has recently received considerable attention for circumventing the drawbacks of VQAs.
It is another type of quantum-classical hybrid algorithm, positioned between the NISQ device and the fault-tolerant quantum computer, that leverages the quantum circuit to project a large-scale Hamiltonian onto a lower-dimensional space, known as the Krylov subspace. Due to its reduced dimension, a classical computer then solves the generalized eigenvalue problem (GEVP) within this subspace.
Existing studies on the QKSD method have utilized various types of the Krylov bases that rely on the base operator and the reference state. 
These include real-time evolution-based QKSD \cite{quantum_filter_method}, the quantum power method (QPM), whose basis is generated by using the Hamiltonian operator \cite{quantum_power_method}, and the QLanczos method, which utilizes imaginary-time evolution \cite{quantum_imaginary_time_evolution}.
Several other approaches, including the use of multiple reference states \cite{doi:10.1021/acs.jctc.9b01125} and multifidelity estimation protocol \cite{PhysRevA.105.022417} have been proposed to improve the QKSD method \cite{PRXQuantum.3.020323, Cohn.2021,Cortes.2022b,Shen.2022,Stair.2022}.

All these methodologies necessarily involve extracting matrix elements of the Krylov subspace for the Hamiltonian through measurements. 
Due to the probabilistic nature of quantum regimes, repeating quantum measurements on multiple copies of the quantum circuit is inevitable for attaining accurate outcomes and probabilities. 
The ever-present error arising from the finite number of measurements cannot be overlooked in the QKSD method, as the GEVP is vulnerable to errors \cite{Theory_QSDK,Mathias2004TheDG,10.2307/2156670}.
The vulnerability of the GEVP has been highlighted in a previous work that conducts an error analysis of quantum subspace diagonalization \cite{Theory_QSDK}.
In this study, we theoretically analyze the finite sampling error and the error of the QKSD method. 
In the QKSD method, we use the Hadamard test and real-time evolution to measure the entries of the projected Hamiltonian and the overlap matrix.
The finite sampling error from this Hadamard test is modeled as Gaussian noise, which causes perturbations in the GEVP.
Based on this error analysis, we employ a threshold, also called a "truncation point,"
to bypass the ill-conditioning of the overlap matrix in the QKSD method.
Finally, we implement the QKSD algorithm to solve the one-dimensional Fermi-Hubbard model by using a numerical approach under the assumption that the quantum circuit used for simulation does not includes any other errors.
We then estimate its eigenenergy perturbation with respect to the number of shots to verify the analysis.

This paper is organized as follows. First, a preliminary version of the QKSD for a Hamiltonian is explained in Section \ref{sec:QKSD} and then a brief overview of QKSD-based error analyses is provided in Section \ref{sec:GEVP}.
The main results are presented in Section \ref{sec:main_result}, which includes the elaborate model of sampling errors from the Hadamard test, the effect of errors on the final results of the GEVP, and optimal thresholding criteria for alleviating the perturbation effect.
Finally, our result is numerically verified with the 1D Fermi-Hubbard model in Section \ref{sec:result}.

\section{QKSD}\label{sec:QKSD}
The Krylov subspace method is an algorithm for evaluating the approximated solutions to high-dimensional matrix problems by projecting large original matrices onto a small space.
In particular, Krylov subspace diagonalization (KSD) solves the eigenvalue problem of the projected matrix to which the original matrix is converted.
If one adopts KSD to solve a quantum many-body problem, the problem matrix corresponds to the many-body Hamiltonian, $\hat{H}$. Furthermore, the order-$n$ Krylov subspace, $\mathcal{K}_n$, is expressed as follows:
\begin{equation}
\label{eq:krylov_subspace}
\mathcal{K}_n\left(\hat{A}, \ket{\phi_0}\right) := \textrm{span}\left(\{\ket{\phi_0}, \ket{\phi_1} \cdots \ket{\phi_{n-1}}\}\right),
\end{equation}
where the nonorthogonal Krylov basis, $\ket{\phi_k} := \hat{A}^{k}\ket{\phi_0}$ is generated by repeatedly applying the base operator, $\hat{A}$, up to $n-1$ times to the reference state vector, $\ket{\phi_0}$.
Generally, the base operator $\hat{A}$ is a function of $\hat{H}$ and is often chosen to be the original Hamiltonian, such as $\hat{A}=\hat{H}$, in classical KSD.
Furthermore, the order $n$ represents a relatively small number compared to the size of a problem ($n\ll\dim \hat{H}$), allowing the classical GEVP solver to address matrices of size $n\times n$.
In Krylov subspace, an ansatz state approximating an eigenstate of $\hat{H}$, $\ket{\psi(\bm{c})}$, is expressed as a linear combination of Krylov basis states:
\begin{equation}
\label{eq:original_state}  
\ket{\psi(\bm{c})} = \frac{1}{\mathcal{N}_{\bm{c}}}\sum_{k=0}^{n-1} {c}_{k} \ket{\phi_k}, 
\end{equation}
where $\mathcal{N}_{\bm{c}}=(\sum_{kl}{c}_{k}^*{c}_{l}\langle\phi_k|\phi_l\rangle)^{1/2}$ is a normalization factor, and $\bm{c}=(c_0, \cdots, c_{n-1})$ represents the unnormalized expansion coefficients.
 When Eq.\eqref{eq:original_state} is viewed as a variational state parameterized by a vector $\bm{c}$, we can adopt the following variational principle:
\begin{equation*}
\begin{split}
\min_{\bm{c}\neq \bm{0}}\braket{\psi(\bm{c})|\hat{H}|\psi(\bm{c})} =& \min_{\bm{c}\neq \bm{0}}\frac{\sum_{kl}c_k^{*}c_l \braket{\phi_k|\hat{H}|\phi_l}}{\sum_{kl}c_k^{*}c_l \braket{\phi_k|\phi_l}}\\
=&\min_{\bm{c}\neq \bm{0}}\frac{\bm{c}^{\dagger}\bm{H}\bm{c}}{\bm{c}^{\dagger}\bm{S}\bm{c}},
\end{split}
\end{equation*}
which corresponds to the generalized Rayleigh quotient of the projected Hamiltonian matrix, $\bm{H}$, and the overlap matrix, $\bm{S}$.
These matrices are defined respectively as
\begin{align}
\bm{H}_{kl} &:= \bra{\phi_k}\hat{H}\ket{\phi_l}, \label{eq:prj_H} \\ 
\bm{S}_{kl} &:= \braket{\phi_k|\phi_l} \label{eq:overlap}.
\end{align}
Therefore, one can subsequently derive a GEVP along with a matrix pair $(\bm{H}, \bm{S})$ \cite{parlett1980symmetric}:
\begin{equation}
\label{eq:gen_eigeq}
\bm{H}\bm{c_j} = \bm{S}\bm{c_j} E^{(n)}_j.
\end{equation}
Here, $E^{(n)}_j$ is the $j^{\text{th}}$ approximated eigenvalue of Krylov subspace order $n$ ($E^{(n)}_0 \le E^{(n)}_1 \le \cdots \le E^{(n)}_{n-1}$), and $\bm{c_j}$ becomes the corresponding $n$-dimensional eigenvector. Subscript $j$ is dropped for the lowest energy and the corresponding eigenvector for convenience.

\begin{figure}[t]
	\centering
	\includegraphics[width=0.98\linewidth]{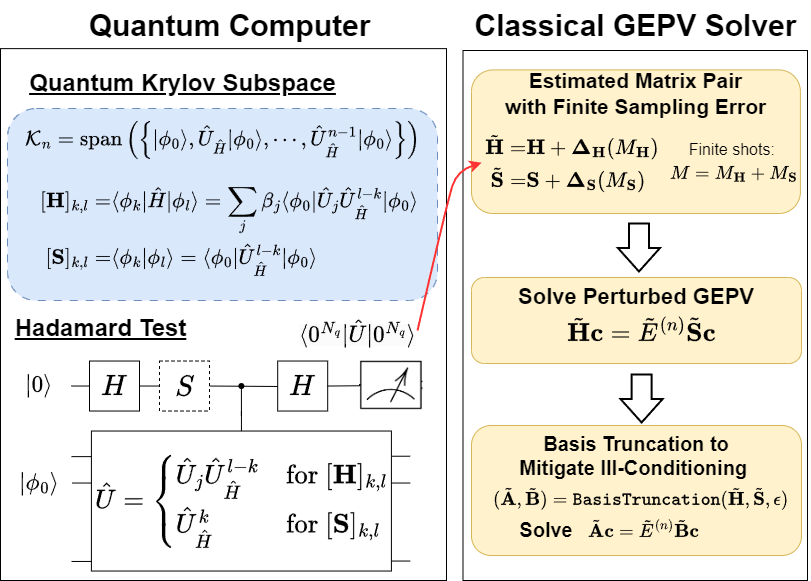}
	\caption{QKSD algorithm involves a quantum part carrying out the projection from the full Hilbert space to the Krylov subspace and a classical part solving the corresponding GEVP.
	Each element in the matrix pair $(\bm{H}, \bm{S})$ is evaluated by the Hadamard test where the controlled $\hat{U}$ gate includes $\hat{U}_{\hat{H}}=e^{-i\hat{H}\Delta_t}$ and $\hat{U}_j$ forming linear combination of unitaries (LCU) of the Hamiltonian.
	Usually, because the GEVP is ill-conditioned, singular vectors with singular values smaller than a parameter $\epsilon$ are discarded by the $\mathtt{BasisThresholding}$ algorithm (see Algorithm \ref{alg:thresholding} in Section \ref{sec:GEVP_PERT}).
	}
	\label{fig:method_flowchart}
\end{figure}

Compared with classical KSD, QKSD is leveraged by quantum simulation algorithms to build a Krylov basis in quantum computers with a reduced level of the computational effort.
Specifically, the quantum algorithms are expected to efficiently yield Eqs.\eqref{eq:prj_H} and \eqref{eq:overlap} within the small error bound, despite the exponentially increasing dimension of the Hilbert space.
Various QKSD algorithms have been proposed for each quantum simulation algorithm, including real-time evolution-based QKSD \cite{quantum_filter_method}, QPM \cite{quantum_power_method}, and quantum imaginary-time evolution (QITE) \cite{quantum_imaginary_time_evolution}.
They are differentiated by the choice of base operator $\hat{A}$, which determines the Krylov subspace.
In the QPM \cite{quantum_power_method}, the base operator $\hat{A}$ is chosen as $\hat{H}$, and its basis, $\hat{H}^k\ket{\phi_0}$, is approximately calculated using the finite difference method (FDM).
However, it exponentially amplifies the sampling errors with respect to $k$, especially when the FDM time step is small, and requires an exponential number of quantum computer calls.
In QITE \cite{quantum_imaginary_time_evolution}, the Krylov subspace is constructed by using an imaginary time evolution basis ($e^{-k\Delta_t}\ket{\phi_0}$) based on the sequential unitary approximations to a normalized imaginary time evolution operator.
However, such approximations involve iterative local tomography, which requires exponential computational effort to expand the range of tomography and compensate for approximation errors.
Therefore, in this work, we focused on the most widely discussed method \cite{quantum_filter_method, PRXQuantum.3.020323, Cohn.2021}, real-time quantum Krylov diagonalization, because it is a compact method with a base operator served by a unitary operator that can be implemented in a substantially simper manner in a quantum circuit than other algorithms invoking nonunitary processes.

We considered the order-$n$ quantum Krylov method that uses a real-time evolution basis with time step $\Delta_t$, $\{\ket{\phi_k}:=e^{-i k \Delta_t \hat{H}}\ket{\phi_0}\}_{k=0}^{n-1}$, whose procedure is described in Fig.\ref{fig:method_flowchart}.
As mentioned elsewhere \cite{quantum_filter_method}, it only requires $n$ elements rather than $n^2$ to define $(\bm{H}, \bm{S})$ if one can perform exact Hamiltonian simulations.
Using the commutation relation, one can rewrite $\bm{H}_{kl}$ and $\bm{S}_{kl}$ in terms of the index difference, $k-l$, expressed as:
\begin{equation}
\begin{split}
\bm{H}_{kl} :=& \bra{\phi_0}e^{ik\hat{H}\Delta_t} \hat{H} e^{-il\hat{H}\Delta_t}\ket{\phi_0}
\\=& \bra{\phi_0}\hat{H} e^{-i(l-k)\hat{H}\Delta_t}\ket{\phi_0}, \label{eq:prj_H_QKD}     
\end{split}    
\end{equation}
\begin{equation}
\begin{split}
\bm{S}_{kl} :=& \bra{\phi_0}e^{ik\hat{H}\Delta_t}e^{-il\hat{H}\Delta_t}\ket{\phi_0}
\\=& \bra{\phi_0}e^{-i(l-k)\hat{H}\Delta_t}\ket{\phi_0}. \label{eq:overlap_QKD}    
\end{split}    
\end{equation}
Subsequently, constructing the Toeplitz matrices, $\bm{H}$, becomes simpler; the elements $\bm{H}_{kl} = h_{l-k}$ are obtained from sequence $\{h_k\}_{k=-(n-1)}^{n-1}$ defined as 
\begin{equation}
\begin{split}
h_k :=& \bra{\phi_0}\hat{H} e^{-ik\hat{H}\Delta_t}\ket{\phi_0}\\
=& \sum_{j=1}^{N_\beta}{\beta_j \bra{\phi_0}\hat{U}_{j} e^{-ik\hat{H}\Delta_t}\ket{\phi_0}}, \label{eq:QKD_H_seq}
\end{split}
\end{equation}
where each term can be obtained from the Hadamard test.
The summation in Eq.\eqref{eq:QKD_H_seq} indicates the separated measurements of $\hat{H}$ by applying unitary partitioning \cite{unitary_partitioning, LCU_Love} with a set of unitaries, $\{\hat{U}_j\}$:
\begin{equation}
\label{eq:pauli_ham_grp}
\hat{H} = \sum_{j=1}^{N_\beta} \beta_j \hat{U}_j,
\end{equation}
where $\beta_j$ are coefficients determined by the unitary partitioning.
Further details are provided in Appendix \ref{sec:appendix_hamiltonian_partitioning}.
Unitary partitioning offers an advantage over the individual measurements of Pauli operators in terms of sampling noise \cite{LCU_Love}.
As an alternative to Eq.\eqref{eq:pauli_ham_grp}, one can adopt fragments of the Hamiltonian, $\hat{H}=\sum_j \hat{H}_j$, where each $\hat{H}_j$ can be diagonalized by an efficiently constructed unitary.
This approach involves using the extended swap test \cite{quantum_filter_method}.
Similarly, matrix $\bm{S}$ is generated by using sequence $\bm{S}_{kl}=s_{l-k}$, where
\begin{equation}
s_k := \bra{\phi_0} e^{-ik\hat{H}\Delta_t}\ket{\phi_0} \label{eq:QKD_S_seq}.
\end{equation} 
Accordingly, we must obtain sequences $\{h_k\}_{k=0}^{n-1}$ and $\{s_k\}_{k=0}^{n-1}$ to fill the matrices using the help of time reversal properties, $h_{-k} = h_{k}^{*}$ and $s_{-k} = s_{k}^{*}$. 

However, the Toeplitz construction is impractical when an approximated simulation, such as Suzuki-Trotter decomposition, possesses a large amount of error.
This issue is inevitable when dealing with Hamiltonians containing noncommuting terms with a shallow quantum circuit.
Here, let us call $\hat{U}_{\mathrm{ST}}(t)=e^{-it\hat{H}}+\hat{\mathcal{E}}_{\mathrm{ST}}(t)$ a decomposition of $e^{-it\hat{H}}$, where $\hat{\mathcal{E}}_{\mathrm{ST}}(t)$ denotes the simulation error.
For example, by adapting the fractal Suzuki-Trotter decomposition \cite{fractal_decomposition, Suzuki_book} of order $r$, the error becomes $\|\hat{\mathcal{E}}_{\mathrm{ST}}(k\Delta_t)\|=O(k\Delta_t^{r+1})$ at a circuit depth of $O(k 2^r N_\Gamma)$, where $k$ is the number of repetitions of $\Delta_t$, and $N_\Gamma$ is the number of terms that can be efficiently diagonalized and thus exponentiated.
When such deep quantum circuits are not accessible, $\hat{\mathcal{E}}_{\mathrm{ST}}$ becomes nonnegligible; thus, an alternative basis, $\{\ket{\phi_k'}:=\hat{U}_{\mathrm{ST}}(\Delta_t)^k \ket{\phi_0}\}$, should be chosen instead of the exact time evolution basis, $\{\ket{\phi_k}:=e^{-ik\Delta_t \hat{H}}\ket{\phi_0}\}$, for the analysis.
In such settings, the Toeplitz nature of $\bm{H}$ (see Eq.\eqref{eq:prj_H_QKD}) vanishes because $[\hat{H},\hat{U}_{\mathrm{ST}}(t)]\neq0$, requiring that $n^2$ elements be obtained instead of $n$.
On the other side, the overlap matrix $\bm{S}$ maintains its Toeplitz property because $\hat{U}_{\mathrm{ST}}(\Delta_t)$ is still unitary, requiring only $n$ elements.

\section{General Error Analysis}\label{sec:GEVP}

The total error of the QKSD algorithm in practice can be attributed to two factors: one is the perturbation of the pair $(\bm{H}, \bm{S})$ owing to the noise from the practical implementation.
As a representative example of a noise source, sampling noise becomes dominant when the elements are obtained using a limited number of samplings; therefore, the finite sampling issue will be mainly addressed in this study.
Conversely, the other factor corresponds to the difference between the approximation solution obtained from a noiseless KSD and the exact solution.
Because the Krylov projection approximates smaller spaces, a projection error can be induced, particularly when the Krylov order $n$ is small.
In addition, the projection error can increase when the overlap between the reference vector $\ket{\phi_0}$ and the true eigenvector is small.
To estimate the lowest eigenvalue of the Hamiltonian $\hat{H}$, the practical execution of QKSD can produce a noisy outcome $\tilde{E}^{(n)}$, causing the total error to be described as a triangular inequality:
\begin{equation}
\label{eq:krylov_total_error}
\left| \tilde{E}^{(n)} - E \right| \le \left| \tilde{E}^{(n)} - E^{(n)} \right| + \left| E^{(n)} - E \right|.
\end{equation}
Here, $E$ is the lowest exact eigenvalue of the Hamiltonian (i.e., the lowest solution of Schr\"{o}dinger's equation), $E^{(n)}$ is the lowest noiseless Krylov approximated eigenvalue of order $n$, and the first and second terms on the right-hand side of Eq.\eqref{eq:krylov_total_error} are two error factors. In the following two subsections, we review recent works on these two error factors.

\subsection{Generalized Eigenvalue Perturbation}\label{sec:GEVP_PERT}
Compared with the eigenvalue problem, the generalized eigenvalue of a matrix pair $(\bm{H}, \bm{S})$ tends to be more sensitive to noise.
In QKSD, this noise can generally involves finite sampling, Suzuki-Trotterization, and other algorithm-dependent errors. 
For each element of the matrices in Eqs.\eqref{eq:prj_H} and \eqref{eq:overlap}, the errors result in matrix perturbations $\bm{\Delta_H}$ and $\bm{\Delta_S}$ as small deviations.
As demonstrated previously \cite{STEWART197969, SUN1982331, 10.1007/BFb0062105}, these perturbations can affect the solution of GEVP, expressed as:
\begin{equation}
\label{eq:perterb_gen_eigeq}
\tilde{\bm{H}} \tilde{\bm{c_j}} = \tilde{\bm{S}} \tilde{\bm{c_j}} \tilde{E}^{(n)}_j, 
\end{equation}
where $\tilde{E}^{(n)}_j = E^{(n)}_j + \Delta E^{(n)}_j$ is the $j^{\text{th}}$ perturbed eigenvalue, and $\tilde{\bm{c_j}}$ is corresponding perturbed eigenvector of the perturbed matrix pair $(\tilde{\bm{H}}:=\bm{H}+\bm{\Delta_H}, \tilde{\bm{S}}:=\bm{S}+\bm{\Delta_S})$.

A previous study by Mathias and Li \cite{Mathias2004TheDG} reported an improved perturbation theory from a geometrical perspective on the complex plane describing the quadratic form of the problem, $\bm{x}^{\dagger}(\bm{H}+i\bm{S})\bm{x}$, for any normalized vector $\bm{x}$.
Subsequently, Epperly et al. \cite{Theory_QSDK} utilized perturbation theory to describe QKSD perturbation with a real-time evolution ansatz.
They also proposed the thresholding that cuts off the least-significant singular values of matrix $\bm{S}$ to address the large condition number.
Because solving the GEVP involves the calculation of $\bm{S}^{-1/2}$, small singular values of $\bm{S}$ may amplify the noise in the matrix pair significantly.
Such cases in which $\bm{S}$ possesses small singular values are called \textit{ill-conditioned problems}.

\begin{figure}[t]
	\begin{algorithm}[H]
		\caption{Truncation of ill-conditioned basis of $\tilde{\bm{S}}$.}
		\label{alg:thresholding}
		\begin{algorithmic}[1]
			\Procedure{BasisThresholding}{$\tilde{\bm{H}}, \tilde{\bm{S}}, \epsilon$}
			\State $n \gets \texttt{dim}(\tilde{\bm{H}})$
			\State $\tilde{\bm{V}}, \tilde{\bm{\Lambda}} \gets \texttt{SVD}(\tilde{\bm{S}})$
			\State $\tilde{\bm{V}}_{>\epsilon} \gets \texttt{EmptyMatrix}(\texttt{size}=(n, 0))$
			\State $j=0$
			\For {$i \gets 0 \cdots n-1$}
			\If {$\tilde{\bm{\Lambda[i]}} > \epsilon$}
			\State $\tilde{\bm{V}}_{>\epsilon}.\texttt{AddColumn()}$
			\State $\tilde{\bm{V}}_{>\epsilon}[:,j] \gets \tilde{\bm{V}}[:,i]$
			\State $j \gets j + 1$
			\EndIf
			\EndFor
			\State $\tilde{\bm{A}}, \tilde{\bm{B}} \gets \tilde{\bm{V}}_{>\epsilon}^\dagger \tilde{\bm{H}}\tilde{\bm{V}}_{>\epsilon}, \tilde{\bm{V}}_{>\epsilon}^\dagger \tilde{\bm{S}}\tilde{\bm{V}}_{>\epsilon}$
			\State \Return $(\tilde{\bm{A}}, \tilde{\bm{B}})$
			\EndProcedure
		\end{algorithmic}
	\end{algorithm}
\end{figure}

Algorithm \ref{alg:thresholding} describes the detailed thresholding technique conducted by removing the singular basis of $\bm{S}$ with smaller singular values than a certain noise level, $\epsilon$.
Such process produces matrices $(\tilde{\bm{A}}, \tilde{\bm{B}}):=\mathtt{BasisThresholding}(\tilde{\bm{H}}, \tilde{\bm{S}}, \epsilon)$ with the size of ${n_\epsilon} \times {n_\epsilon}$ so that the overlap matrix $\tilde{\bm{B}}$ has singular values larger than $\epsilon$.
Therefore, its inverse matrix becomes numerically stable, although the dimensions of the problem are reduced (${n_\epsilon \le n}$).

The following generalized eigenvalue perturbation theory \cite[Corollary~3.6]{Mathias2004TheDG} was introduced to explain the perturbation between matrix pairs $(\tilde{\bm{A}}, \tilde{\bm{B}})$ and $(\bm{A}, \bm{B}):=\mathtt{BasisThresholding}(\bm{H}, \bm{S}, \epsilon)$, where the dimensions of both pairs are assumed to be the same as $n_{\epsilon}$.
Here, the spectral norm is denoted with the symbol $\|\cdot\|$.

\begin{figure}[t]
	\centering
	\includegraphics[width=0.8\linewidth]{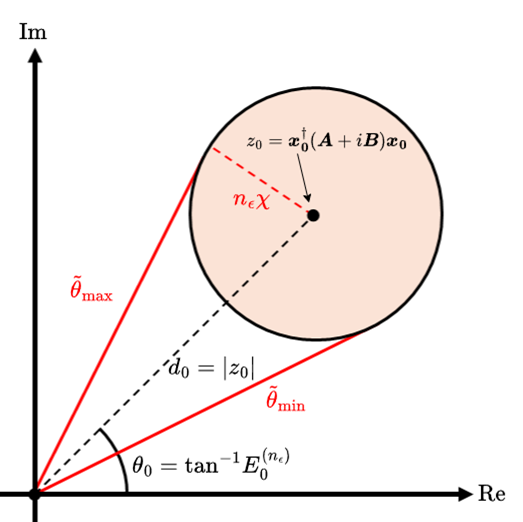}
	\caption{
		A geometrical description of Theorem \ref{theorem:generalized_eigenvalue_perturbation} (\cite[Corollary 3.6]{Mathias2004TheDG}).
		For the $n_\epsilon$-dimensional GEVP of ($\bm{A}, \bm{B}$) and a normalized eigenvector $\bm{x_0}$, a complex value $z_0=\bm{x_0}^\dagger(\bm{A}+i\bm{B})\bm{x_0}$ is defined.
		Then, the corresponding eigenangle is defined and determined as $\theta_0 := \mathrm{arg}(z_0) = \mathrm{tan}^{-1}E_0^{(n_\epsilon)}$, where $E_0^{(n_\epsilon)}$ is the eigenvalue.
		If the perturbation is induced, which respectively shifts ($\bm{A}, \bm{B}$) and $\bm{x_0}$ to ($\tilde{\bm{A}}, \tilde{\bm{B}}$) and $\tilde{\bm{x}}_{\bm{0}}$, one can show that the perturbation on $z_0$, $\tilde{z}_0:=\tilde{\bm{x}}_{\bm{0}}^\dagger(\tilde{\bm{A}}+i\tilde{\bm{B}})\tilde{\bm{x}}_{\bm{0}}$ is located within the disk of $|\tilde{z}_0 - z_0|\le n_\epsilon\chi$, where $\chi$ is defined as Eq.\eqref{eq:def_chi}.
		Therefore, the corresponding perturbed eigenangle is bounded within $[\tilde{\theta}_{\mathrm{min}}, \tilde{\theta}_{\mathrm{max}}]$, where $|\theta_0 - \tilde{\theta}_{\mathrm{min}}|=|\theta_0 - \tilde{\theta}_{\mathrm{max}}|=\mathrm{sin}^{-1}(n_\epsilon\chi / d_0)$ is found with the geometrical interpretation.
	}
	\label{fig:fig_pert_thm}
\end{figure}

\begin{theorem}[Generalized Eigenvalue Perturbation Theory, Corollary 3.6 of a previous report \cite{Mathias2004TheDG}.]
	\label{theorem:generalized_eigenvalue_perturbation}
	For the $n_\epsilon\times n_\epsilon$ matrix pair $(\bm{A}, \bm{B})$, whose eigenvalues are $E_0^{(n_\epsilon)}\le E_1^{(n_\epsilon)} \cdots \le E_{n_\epsilon-1}^{(n_\epsilon)}$, with positive definite $\bm{B}$, assume that the matrix pair $(\tilde{\bm{A}}, \tilde{\bm{B}})=(\bm{A}+\bm{\Delta_A}, \bm{B}+\bm{\Delta_B})$ is perturbed from $(\bm{A}, \bm{B})$, and $\chi$ is defined as
	\begin{equation}\label{eq:def_chi}
	\chi := \sqrt{\|\bm{\Delta_A}\|^2+\|\bm{\Delta_B}\|^2}.
	\end{equation}
	Assume that error bound $\chi$ is sufficiently small,
	\begin{equation}\label{eq:assume_err_bound}
	\sqrt{2}{n_\epsilon}\chi \le \lambda_{n_\epsilon},
	\end{equation}
	and the following gap condition holds:
	\begin{equation}\label{eq:assume_gap_condition}
	|\tan^{-1}E_1^{({n_\epsilon})} - \tan^{-1}E_0^{({n_\epsilon})}| \ge \sin^{-1}\frac{{n_\epsilon}\chi}{\lambda_{n_\epsilon}},
	\end{equation}
	where $\lambda_{n_\epsilon}$ is the smallest singular value of $\bm{B}$.
	Then, the lowest perturbed eigenangle of $(\tilde{\bm{A}}, \tilde{\bm{B}})$, $\tan^{-1}\tilde{E}_0^{({n_\epsilon})}$ satisfies the following bound:
	\begin{equation}
	\label{eq:gep_theorem_result}
	\left| \tan^{-1}E^{({n_\epsilon})}_0 - \tan^{-1}\tilde{E}^{({n_\epsilon})}_0 \right| \le \sin^{-1}\frac{{n_\epsilon}\chi}{d_0}.
	\end{equation}
	Here, $d_0^{-1}$ is the condition number of the eigenangle $\tan^{-1}E_0^{({n_\epsilon})}$, defined as
	\begin{equation}
	d_0^{-1} := |\bm{x}_0^{\dagger}(\bm{A}+i\bm{B})\bm{x}_0|^{-1},
	\end{equation}
	where $\bm{x}_0$ is the unit-norm eigenvector of $(\bm{A}, \bm{B})$ with the lowest eigenvalue.
\end{theorem}

A pictorial representation of Theorem \ref{theorem:generalized_eigenvalue_perturbation} is shown in Fig.\ref{fig:fig_pert_thm}.
Note that the smallest eigenvalue of $\bm{B}$ is assured to be larger than $\epsilon$ ($\lambda_{\epsilon}>\epsilon$) because $(\bm{A}, \bm{B})$ is thresholded from $(\bm{H}, \bm{S})$.
So, with a sufficiently small $\chi$, the assumption in Eq.\eqref{eq:assume_gap_condition} may accept relatively small gap between the eigenangles.
Although Theorem \ref{theorem:generalized_eigenvalue_perturbation} may explain the perturbation of $(\tilde{\bm{A}}, \tilde{\bm{B}})$ from the ideal $(\bm{A}, \bm{B})$ obtained by truncating the basis of $(\bm{H}, \bm{S})$ using Algorithm \ref{alg:thresholding}, the error magnitude $\chi$ could not still be explicitly obtained in terms of the error matrices $(\bm{\Delta_H}, \bm{\Delta_S})$.
Hence, another error magnitude is defined to perform error analysis using explicit matrices:
\begin{equation}
\label{eq:def_eta}
\eta := \sqrt{\|\bm{\Delta_H}\|^2+\|\bm{\Delta_S}\|^2}.
\end{equation}
Subsequently, the relationship between the explicit and implicit error magnitudes, $\eta$ and $\chi$, is obtained by the following \cite[Theorem~2.7]{Theory_QSDK}, which is
\begin{equation}
\label{eq:chi_bound}
\chi \le O\left(\eta^{\frac{1}{1+\alpha}}/n\right),
\end{equation}
subject to the condition
\begin{equation}
\label{eq:asymptotic_epsilon}
\epsilon = \Theta \left( \eta^{\frac{1}{1+\alpha}} \right),
\end{equation}
where $\alpha$ is a constant ranging from 0 to $1/2$.
Finally, the corresponding asymptotic perturbation bound is given by
\begin{equation}
\left| \tan^{-1}E^{(n\rightarrow {n_\epsilon})}_0 - \tan^{-1}\tilde{E}^{(n\rightarrow {n_\epsilon})}_0 \right| \le O\left(\eta^{\frac{1}{1+\alpha}}d_0^{-1} \right).
\label{eq:informal_eigenangle_error}
\end{equation}
Here, the superscript $(n\rightarrow {n_\epsilon})$ denotes the calculated eigenvalue obtained from the ${n_\epsilon}$-dimensional subspace produced by a thresholding from the $n$-dimensional space.

In summary, the perturbation bound in Eq.\eqref{eq:informal_eigenangle_error} indicates that the perturbation error is sublinear to the error matrix norms and condition number after the truncation of the basis.
Thus, with additional information about the error matrix norms $\|\bm{\Delta_H}\|$ and $\|\bm{\Delta_S}\|$, one can establish a sampling error analysis for the QKSD algorithm.
Appendix \ref{sec:appendix_generalized_eignevalue_perturbation} presents a detailed version \cite[Theorem~2.7]{Theory_QSDK}.

\subsection{Krylov Approximations}
Concerning the effect of projection to the small Krylov subspace, Krylov convergence, represented as the second factor on the right-hand side of Eq.\eqref{eq:krylov_total_error}, is reviewed here.
For the classical version, this effect was quantified \cite{10.2307/2156670} by plugging an ansatz based on Chebyshev polynomials.
In contrast, Stair et al. \cite{doi:10.1021/acs.jctc.9b01125} argued that the classical ($\hat{A}=\hat{H}$) and quantum ($\hat{A}=e^{-i\hat{H}\Delta_t}$) Krylov space could be close, up to $O(\Delta_t^2)$.
Their work allows the utilization of the classical result \cite{10.2307/2156670} for the analysis of QKSD when a small $\Delta_t$ is employed.
However, a small $\Delta_t$ to make the QKSD similar to the classical KSD causes ill-conditioning and degrades the perturbation error.
A remarkable QKSD error analysis was proposed by Epperly et al. \cite{Theory_QSDK}, which is advantageous compared with previous studies in two ways: the thresholding effect is considered, and it is directly analogous to classical KSD analysis \cite{10.2307/2156670} by replacing Chebyshev polynomials with trigonometric polynomials.
Their analysis is reviewed in the following.

\begin{figure*}[t!]
	\centering
	\includegraphics[width=\linewidth]{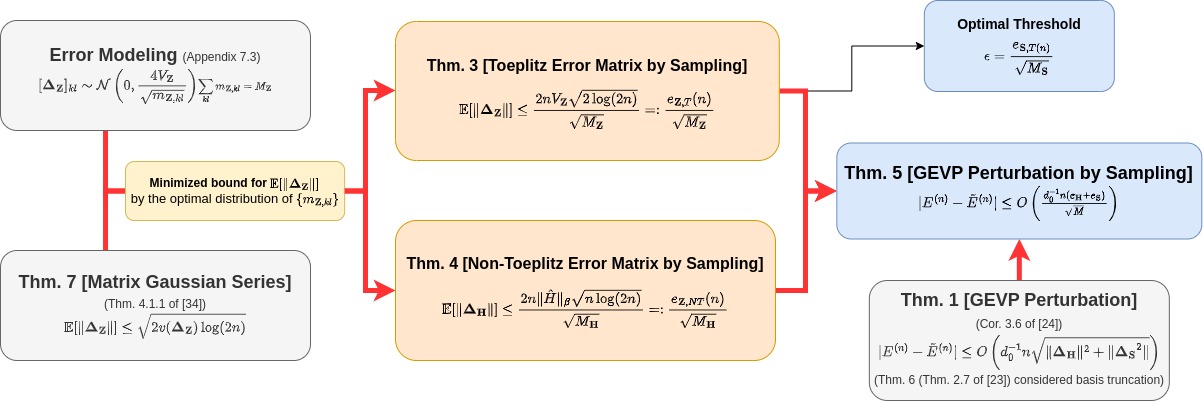}
	\caption{Overview of the main results.
		Error matrices are considered as Gaussian random matrices whose standard deviations are derived in Appendix \ref{sec:appendix_hadamard_sampling_variance}.
		The statistics of the error matrix norm is derived using \cite[Theorem~4.4.1]{MAL-048}, whose proofs are provided in Appendix \ref{sec:appendix_matrix_variance_statistics}.
		Accordingly, a bound of the expected norm is estimated with the optimal distribution of shots in Theorems \ref{theorem:toeplitz_error_matrix} and \ref{theorem:non_toeplitz_error_matrix} for each case of Toeplitz and non-Toeplitz construction.
		Based on the previously known GEVP perturbation theories \cite[Corollary~3.6]{Mathias2004TheDG} and \cite[Theorem~2.7]{Theory_QSDK}, we derive the GEVP perturbation caused by finite sampling.
		Additionally, optimal threshold parameter $\epsilon$ is suggested using Theorem \ref{theorem:toeplitz_error_matrix}.}
	\label{fig:error_analysis_flowchart}
\end{figure*}

\begin{theorem}[Quantum Krylov Error, Theorem 1.2 of  \cite{Theory_QSDK}.]\label{theorem:qksd_err}
	Let $\ket{\psi_0}\cdots \ket{\psi_{N-1}}$ and $E_0\le\cdots\le E_{N-1}$ be the eigenvectors and corresponding eigenvalues, respectively, for a Hamiltonian $\hat{H}$ in the Hilbert space (dimension: $N$). Then, the initial state is expanded on an eigenbasis,
	\begin{equation}
	\ket{\phi_0} = \sum_{j=0}^{N-1}\gamma_i\ket{\psi_i}.
	\end{equation}
	Suppose that the Krylov subspace is generated within time sequence $\{t_j:=\pi j/\Delta E_{N-1}\}_{j=-k}^k$ for a certain integer, $k$; then, the Krylov order can be given as $n=2k+1$, where $\Delta E_j := E_j -E_0$ is also defined.
	Overall, if the GEVP returned by Algorithm \ref{alg:thresholding} with the parameter $\epsilon$ in Eq.\eqref{eq:asymptotic_epsilon} is solved, the following upper bound for the error between the lowest eigenangle and the full Hilbert space solutions holds:
	\begin{multline}  \label{eq:krylov_convergence}
	\left|\tan^{-1}E_0^{(n)} - \tan^{-1} E_0\right|\\ \le O\left( \frac{1-|\gamma_0|^2}{|\gamma_0|^2} e^{-n O\left( \frac{\Delta E_1}{\Delta E_{N-1}} \right)} + \frac{\Delta E_{N-1}}{|\gamma_0|^2}\eta^{\frac{1}{1+\alpha}} \right).      
	\end{multline}
\end{theorem}

In the right-hand side of Eq.\eqref{eq:krylov_convergence}, the first term corresponds to the approximation error between the full Krylov subspace and the Hilbert space, whereas the second term corresponds to the subspace after basis truncation.
Furthermore, if the effect of thresholding is small, it implies that the Krylov error bound decays exponentially to the Krylov order, with the prefactor being inversely proportional to the overlap, $|\gamma_0|^2$. 

\section{Main Results}\label{sec:main_result}

When evaluating the elements of matrices $\bm{H}$ in Eq.\eqref{eq:QKD_H_seq} and $\bm{S}$ in Eq.\eqref{eq:QKD_S_seq}, the Hadamard test subroutine intrinsically suffers from sampling errors, even if the  other noises such as hardware noise and Trotter error are reduced to a negligible level.
Although Eq.\eqref{eq:krylov_convergence} shows that the QKSD algorithm converges quickly toward the full Hilbert space solution, it is necessary to assess the tolerance to perturbation errors resulting from finite sampling.
Thus, in this section, Krylov perturbations are mainly investigated with respect to a finite number of shots.
Fig.\ref{fig:error_analysis_flowchart} illustrates an overview of the structure of this section.
First, the sampling error model using a Gaussian random matrix is described, i.e., the element in the Gaussian random matrix is regarded as the sampling error.
The variance of the Gaussian element is set from that of the binomial distribution whose random variable is the averaged outcome obtained by implementing a finite number of Hadamard tests. 
Then, the statistical behavior of the error matrix norms, $\bm{\Delta_H}$ and $\bm{\Delta_S}$, is formulated.
Last, the norm statistics and Theorem \ref{theorem:generalized_eigenvalue_perturbation} are subsequently used to suggest a nearly optimal threshold to mitigate the effect of ill-conditioning bases and finally derive the sampling perturbation bound of the QKSD.

\subsection{Error Model}
Random matrix theory is introduced to address the behavior of random matrices whose elements consist of sampling errors.
Because the single-shot outcome of Hadamard test is binary, the sampling average follows a binomial distribution.
For example, if the real part of $\bm{S}_{kl}$ is evaluated using $m_{\bm{S},kl}^{(\mathrm{r})}$ samplings, its noisy element follows a binomial distribution $\textrm{Re}[\tilde{\bm{S}}_{kl}]\sim 2\textrm{Bin}\left(m_{\bm{S},kl}^{(\mathrm{r})}, p^{(\mathrm{r})}_{kl}\right)/m_{\bm{S},kl}^{(\mathrm{r})}-1$, where $p^{(\mathrm{r})}_{kl}$ denotes the true probability that exactly estimates the matrix element ($2p^{(\mathrm{r})}_{kl}-1=\textrm{Re}[\bm{S}_{kl}]$).
Subsequently, if $m_{\bm{S},kl}^{(\mathrm{r})}$ is sufficiently large, $\mathrm{Re}[\tilde{\bm{S}}_{kl}]$ can be approximated by a Gaussian distribution.
Therefore, if we consider the sampling of the imaginary part in the similar way, the off-diagonal elements in the first row ($[\bm{\Delta_S}]_{0l}$ for $l>0$) are expressed as complex Gaussian random variables, where the real and imaginary parts are independent:
\begin{equation}\label{eq:delta_s}
[\bm{\Delta_S}]_{0l} \sim \mathcal{N}\left(0, \sigma_{\bm{S},0l}^{(\mathrm{r})2}\right) + i\mathcal{N}\left(0, \sigma_{\bm{S},0l}^{(\mathrm{i})2}\right).
\end{equation}
Here, $\sigma_{\bm{S},0l}^{(\mathrm{r})2}$ and $\sigma_{\bm{S},0l}^{(\mathrm{i})2}$ are the variances of $\mathrm{Re}[\tilde{\bm{S}}_{0l}]$ and $\mathrm{Im}[\tilde{\bm{S}}_{0l}]$, respectively.
The other rows are determined by the first one, because the residual matrix $\bm{\Delta_S}$ is Toeplitz and Hermitian.
Specifically, the upper triangular part ($k<l$) is assigned as $[\bm{\Delta_S}]_{kl}=[\bm{\Delta_S}]_{0,l-k}$, while the lower triangular part is its transposed conjugate.
Furthermore, we set the error in the diagonal elements to zero 
\begin{equation}\label{eq:diagonal_S}
[\bm{\Delta_S}]_{kk}=0,
\end{equation} 
because $\bm{S}_{kk}=\braket{\phi_k|\phi_k}=1$ is determined without sampling.

Meanwhile, for the Hamiltonian matrix, unitary fragments are independently measured as shown in Eq.\eqref{eq:QKD_H_seq}.
This fragmentation results in the residual matrix whose upper triangular elements ($k < l$) are
\begin{equation}\label{eq:delta_h}
\begin{split}
[\bm{\Delta}_{\bm{H}}]_{kl} &\sim \sum_{j=1}^{N_\beta}\beta_j\left(\mathcal{N}\left(0, \sigma_{\bm{H},kl;j}^{(\mathrm{r})2}\right) + i \mathcal{N}\left(0, \sigma_{\bm{H},kl;j}^{(\mathrm{i})2}\right)\right)\\
&=\mathcal{N}\left(0, \sigma_{\bm{H},kl}^{(\mathrm{r})2}\right) + i \mathcal{N}\left(0, \sigma_{\bm{H},kl}^{(\mathrm{i})2}\right),
\end{split}
\end{equation}
where $\sigma_{\bm{H},kl;j}^{(\mathrm{r})2}$ and $\sigma_{\bm{H},kl;j}^{(\mathrm{i})2}$ denote the variances associated with measuring real and imaginary part of $\braket{\phi_k|\hat{U}_j|\phi_l}$ respectively, and $\sigma_{\bm{H},kl}^{(\mathrm{r,i})2}:=\sum_j \beta_j^2 \sigma_{\bm{H},kl;j}^{(\mathrm{r,i})2}$ is defined.
Due to the Hermiticity of $\tilde{\bm{H}}$, the lower triangular elements ($k > l$) are assigned as $[\bm{\Delta_H}]_{kl}=[\bm{\Delta_H}]_{lk}^*$.
However, unlike the case of $\bm{\Delta_S}$, the Toeplitz property of $\bm{\Delta_H}$ depends on the construction method.
If $\tilde{\bm{H}}$ is assumed to be Toeplitz, only the elements in the first row ($k=0$) are sampled from the distribution in Eq.\eqref{eq:delta_h}, which then determines the other elements.
Furthermore, because the diagonal elements are real, the Hadamard tests for their imaginary parts are skipped, resulting in the diagonal residual elements of
\begin{equation}\label{eq:diagonal_H}
[\bm{\Delta_H}]_{kk}\sim \mathcal{N}(0, \sigma_{\bm{H},kk}^{(\mathrm{r})2}).
\end{equation}

The variance of each Hadamard test is governed by the standard quantum limit (SQL), which is inversely proportional to the number of shots allocated.
However, since a single matrix element is estimated by multiple Hadamard tests, a proper allocation of shots can minimize the variance within the SQL.
For example, let us assume $m_{\bm{S},0l}=m_{\bm{S},0l}^{(\mathrm{r})}+m_{\bm{S},0l}^{(\mathrm{i})}$ shots are given to measure $[\tilde{\bm{S}}]_{0l}$ (Eq.\eqref{eq:delta_s}), where $m_{\bm{S},0l}^{(\mathrm{r})}$ and $m_{\bm{S},0l}^{(\mathrm{i})}$ are the number of shots to measure real and imaginary part, respectively.
Then, the variance is determined as
\begin{equation}
\sigma_{\bm{S},0l}^{(\mathrm{r})2}+\sigma_{\bm{S},0l}^{(\mathrm{i})2}=\frac{1-\mathrm{Re}[[\bm{S}]_{0l}]^2}{m_{\bm{S},0l}^{(\mathrm{r})}}+\frac{1-\mathrm{Im}[[\bm{S}]_{0l}]^2}{m_{\bm{S},0l}^{(\mathrm{i})}},
\end{equation}
and is minimized by $m_{\bm{S},0l}^{(\mathrm{r})}\propto (1-\mathrm{Re}[[\bm{S}]_{0l}]^2)^{1/2}$ and $m_{\bm{S},0l}^{(\mathrm{i})}\propto (1-\mathrm{Im}[[\bm{S}]_{0l}]^2)^{1/2}$.
However, such shot allocation cannot be determined in advance because $[\bm{S}]_{0l}$ is not yet identified.
Therefore, instead of using the exact value of $[\bm{S}]_{0l}$, its Haar-averaged value can be used to estimate the shot allocation, which results in $m_{\bm{S},kl}^{(\mathrm{r})} = m_{\bm{S},kl}^{(\mathrm{i})}$.
The optimization procedure for the Hadamard test variance is detailed in Appendix \ref{sec:appendix_hadamard_sampling_variance}.

For the case of Hamiltonian matrix, the shots are allocated to each unitary fragment: $m_{\bm{H},kl}=\sum_{j}(m_{\bm{H},kl;j}^{(\mathrm{r})} + m_{\bm{H},kl;j}^{(\mathrm{i})}$), where $m_{\bm{H},kl;j}^{(\mathrm{r,i})}$ shots are allocated to measure the real and imaginary part of $\braket{\phi_k|\hat{U}_j|\phi_l}$.
In Appendix \ref{sec:hamiltonian_overlap}, we showed that $m_{\bm{H},kl;j}^{(\mathrm{r,i})}\propto |\beta_j|$ is the optimal allocation, assuming the value of $\braket{\phi_k|\hat{U}_j|\phi_l}$ is unknown.

For average-case complexity analysis, we take Haar-averaging over $\ket{\phi_k}$ and $\ket{\phi_l}$ and define the averaged variance for the $(k,l)$ element of matrix $\bm{\Delta_Z}\in\{\bm{\Delta_H}, \bm{\Delta_S}\}$:
\begin{equation}
\sigma_{\bm{\bm{Z}},kl}^2 := \mathbb{E}_{\ket{\phi_k}, \ket{\phi_l}\sim \mathcal{H}^2}\left[\sigma_{\bm{Z},kl}^{(\mathrm{r})2}+\sigma_{\bm{Z},kl}^{(\mathrm{i})2}\right].
\end{equation} 
As a consequence, the averaged sampling variances are obtained as
\begin{align}
\sigma_{\bm{H},kl}^2 =& \frac{2(2-\delta_{kl})\|\hat{H}\|_\beta^2}{m_{\bm{H},kl}} \label{eq:var_mat_elem_H}\\
\sigma_{\bm{S},kl}^2 =& \frac{4(1-\delta_{kl})}{m_{\bm{S},kl}},\label{eq:var_mat_elem_S}
\end{align}
where $\|\hat{H}\|_\beta:=\sum_j |\beta_j|$ is the 1-norm of the unitary partitioning (Eq.\eqref{eq:pauli_ham_grp}).
Also, $\delta_{kl}$ equals to one if $k=l$ or zero elsewhere and appears because diagonal elements are estimated in simpler ways (Eqs.\eqref{eq:diagonal_S} and \eqref{eq:diagonal_H}).
Appendix \ref{sec:appendix_hadamard_sampling_variance} shows the detailed derivations of Eqs.\eqref{eq:var_mat_elem_H} and \eqref{eq:var_mat_elem_S}.

\subsection{Norm Behavior of Error Matrices}
In this section, by employing the aforementioned approach of modeling random matrices, the norm statistical properties of such matrices are derived, as matrix norm is used in the perturbation theorem (Theorem \ref{theorem:generalized_eigenvalue_perturbation}) and Eq.\eqref{eq:def_eta}.
As an instance of random matrix theory, it is well known that the $n\times n$ Gaussian matrix with independent and zero-mean elements has a spectral norm asymptotically growing as fast as $O(\sigma_{max}\sqrt{n})$ with high probability \cite[Theorem~4.4.5]{vershynin_2018}, where $\sigma_{max}$ is the maximum value of the standard deviations of the matrix elements.
However, to obtain a non-asymptotic bound and impose the Hermitian and Toeplitz conditions, we use a slightly different approach \cite[Theorem~4.1.1]{MAL-048}, where the matrix concentration is based on the matrix Laplace transform.

As reviewed in Section \ref{sec:QKSD}, matrices $\bm{H}$ and $\bm{S}$ are constructed using a series of Hadamard tests.
Specifically, while $\bm{S}$ is always constructed as a Toeplitz matrix, the matrix $\bm{H}$ can be chosen as either Toeplitz or non-Toeplitz depending on whether a precise simulation of $e^{-i\hat{H}n\Delta_t}$ is available.
First, one considers an $n\times n$ Hermitian Toeplitz matrix constructed from a sequence of $n$ complex values, or, equivalently, $2n-1$ real values.
By treating these values as complex Gaussian random variables whose variances are described in Eqs.\eqref{eq:var_mat_elem_H} and \eqref{eq:var_mat_elem_S}, one can then find the statistical behavior of $\|\bm{\Delta_Z}\|$.
In addition, the bound on its expectation value, shown in \cite[Theorem~4.1.1]{MAL-048}, is minimized with a proper distribution of the total samplings $M_{\bm{Z}}$ to the Hadamard test of each matrix element.
Appendices \ref{sec:Toeplitz_mat_theory_proof} and \ref{sec:Non_Toeplitz_mat_theory_proof} describe the results and provide detailed proofs of the following theorems.

\begin{theorem}[Toeplitz Error Matrix Obtained by Sampling]
	\label{theorem:toeplitz_error_matrix}
	Supposing that the Toeplitz Hermitian matrix $\bm{Z}\in\{\bm{H}, \bm{S}\}$ is generated by a sequence of Hadamard tests with a given total number of samplings $M_{\bm{Z}}$, the bound for the expected norm of noise matrix $\|\bm{\Delta_Z}^{(\mathrm{t})}\|$, suggested previously \cite[Theorem~4.1.1]{MAL-048}, can be minimized by adequately distributing $M_{\bm{Z}}$ to each Hadamard test setting:
	\begin{align}
	m_0 = m_0^{(\mathrm{r})} &= \frac{M_{\bm{Z}}\delta_{\bm{ZH}}}{\sqrt{2}(n-1)+1}, \\
	m_{k>0}^{(\mathrm{r})}= m_{k>0}^{(\mathrm{i})} &= \frac{M_{\bm{Z}}}{2(n-1)+\sqrt{2}\delta_{\bm{ZH}}}.
	\end{align}
	Here, $m_k$ denotes the number of samplings for the Hadamard test to obtain the $k^{\text{th}}$ superdiagonal and subdiagonal elements, and the superscripts $(\mathrm{r})$ and $(\mathrm{i})$ correspond to the real and imaginary Hadamard test configurations.
	Also, $\delta_{\bm{ZH}}=1$ for the case of $\bm{Z}=\bm{H}$ and zero otherwise.
	Consequently, the optimal expectation bound becomes
	\begin{multline}\label{eq:toeplitz_error_matrix_norm_bound}
	\mathbb{E}[\|\bm{\Delta_Z}^{(\mathrm{t})}\|] \le \frac{2V_{\bm{Z}}(\sqrt{2}n+\delta_{\bm{ZH}}-\sqrt{2})\sqrt{\log(2n)}}{\sqrt{M_{\bm{Z}}}} \\ \lesssim \frac{2nV_{\bm{Z}}\sqrt{2\log(2n)}}{\sqrt{M_{\bm{Z}}}} =: \frac{e_{\bm{Z}}^{(\mathrm{t})}(n)}{\sqrt{M_{\bm{Z}}}}
	\end{multline}
	where $V_{\bm{Z}}$ is the variance prefactor depending on matrix $\bm{Z}$, specifically, $V_{\bm{H}}=\|\hat{H}\|_\beta$ and $V_{\bm{S}}=1$, and $e_{\bm{Z}}^{(\mathrm{t})}(n):=2nV_{\bm{Z}}\sqrt{2\log(2n)}$.
\end{theorem}

Similar to Theorem \ref{theorem:toeplitz_error_matrix}, the non-Toeplitz version of the error matrix statistics is also presented below. However, independent of the construction method, only the matrix $\bm{H}$ is considered here because $\bm{S}$ is always Toeplitz.

\begin{theorem}[Non-Toeplitz Error Matrix Obtained by Sampling]
	\label{theorem:non_toeplitz_error_matrix}
	Assuming that a sequence of Hadamard tests generates the non-Toeplitz Hermitian $\bm{H}$ with a total of $M_{\bm{H}}$ samplings, the bound for the expected norm of the noise matrix $\|\bm{\Delta_H}^{(\mathrm{nt})}\|$, suggested previously \cite[Theorem~4.1.1]{MAL-048}, is minimized by the optimal distribution
	\begin{align}
	m_{kk}^{(\mathrm{r})}=m_{k<l}^{(\mathrm{r})}=m_{k<l}^{(\mathrm{i})}=\frac{M_{\bm{H}}}{n^2},
	\end{align}
	where $m_{kl}^{\mathrm{(r,i)}}$ denotes the number of samplings to evaluate the real or imaginary part of the $(k,l)$ and $(l,k)$ sites of $\bm{H}$. The corresponding optimal bound for the expected norm is
	\begin{equation}\label{eq:non_toeplitz_error_matrix_norm_bound}
	\mathbb{E}[\|\bm{\Delta_H}^{(\mathrm{nt})}\|] \le \frac{2n\|\hat{H}\|_\beta\sqrt{n\log(2n)}}{\sqrt{M_{\bm{H}}}} =: \frac{e_{\bm{H}}^{(\mathrm{nt})}(n)}{\sqrt{M_{\bm{H}}}},
	\end{equation}
	where $e_{\bm{H}}^{(\mathrm{nt})}(n):=2n\|\hat{H}\|_\beta\sqrt{n\log(2n)}$.
\end{theorem}

Under the Toeplitz (and non-Toeplitz) construction, although the number of independent matrix elements is taken as $O(n)$ (and $O(n^2)$), the number of samplings $M_{\bm{Z}}$ should increase as fast as $O(n^2\log{n})$ (and $O(n^3\log{n})$) to maintain the expected error norms below a constant level.
Appendix \ref{sec:appendix_matrix_variance_statistics} reviews the statistical distribution of $\|\bm{\Delta_Z}\|$, known as sub-gaussian from concentration inequality \cite[Theorem~5.6.]{10.1093/acprof:oso/9780199535255.002.0004}.
Defining a bound factor $\kappa>0$, such that $e_{\bm{Z}}(n)/\sqrt{M_{\bm{Z}}}=(1+\kappa)\mathbb{E}[\|\bm{\Delta_Z}\|]$, and applying the concentration inequality then gives
\begin{equation}\label{eq:mat_norm_concentration}
\mathbb{P}\left\{ \|\bm{\Delta_Z}\| \ge \frac{e_{\bm{Z}}(n)}{\sqrt{M_{\bm{Z}}}} \right\} \le \left(\frac{1}{2n}\right)^{\left( 1+\sfrac{1}{\kappa} \right)^{-2}},
\end{equation}
where $e_{\bm{Z}}(n)\in \{e_{\bm{Z}}^{(\mathrm{t})}(n), e_{\bm{Z}}^{(\mathrm{nt})}(n)\}$ is determined by the construction method.
Eq.\eqref{eq:mat_norm_concentration} describes the probability of finding a norm larger than that of the upper bound for the expected norm described in Eqs.\eqref{eq:toeplitz_error_matrix_norm_bound} and \eqref{eq:non_toeplitz_error_matrix_norm_bound}.
The numerical results show that the value of $\kappa$ is located within the interval $[0.5, 1.0]$; however, the actual statistical distribution of $\|\bm{\Delta_Z}\|$ is much tighter than in Eq.\eqref{eq:mat_norm_concentration}, causing all $\|\bm{\Delta_Z}\|$ values with the 10,000 random initializations to be less than $e_{\bm{Z}}/\sqrt{M_{\bm{Z}}}$.
Based on such observation, we assume that the following bound condition holds:
\begin{equation}\label{eq:mat_norm_bound}
\|\bm{\Delta_Z}\| < \frac{e_{\bm{Z}}(n)}{\sqrt{M_{\bm{Z}}}}.
\end{equation}

\subsection{QKSD Perturbation Caused by Finite Sampling}
Finally, by applying Eq.\eqref{eq:mat_norm_bound} to Theorem \ref{theorem:generalized_eigenvalue_perturbation}, we connect the random error matrix statistics to the perturbation theory to estimate the impact of finite sampling.
Before doing so, however, the unstable basis vectors are resolved by examining the singular values of the matrix $\tilde{\bm{S}}$, thereby restoring the generalized eigenvalue problem from ill-conditioning.
Specifically, the singular vectors of $\tilde{\bm{S}}$ are removed if corresponding singular values are below parameter $\epsilon$, as specified in Algorithm \ref{alg:thresholding}.
In practice, $\epsilon$ is set to the noise level, such as machine precision in classical KSD.
Similarly, as shown in Eq.\eqref{eq:mat_norm_bound}, we set the parameter $\epsilon$ for QKSD as
\begin{equation}\label{eq:opt_epsilon}
\epsilon = \frac{e_{\bm{S}}^{(\mathrm{t})}(n)}{\sqrt{M_{\bm{S}}}}.
\end{equation}
Although fewer basis vectors are considered after thresholding, the perturbation error of the final solution is significantly reduced in general.
The numerical experiments in the next section will show that such settings of $\epsilon$ yield lower errors in the final solution than any other assignment to $\epsilon$.
In addition to Eq.\eqref{eq:chi_bound}, which shows the bound of $\chi$ with $\eta$, we assume a simpler bound that removes the dependency on $n$ ($\chi\le K\eta$), where $K>0$ is a factor that is close to unity in the numerical simulation.
Combining these assumptions with Theorem \ref{theorem:generalized_eigenvalue_perturbation}, we can formulate a non-asymptotic bound for the QKSD solution error affected by finite sampling.

\begin{theorem}[Non-Asymptotic Sampling Perturbation after Thresholding]\label{theorem:main_result}
Suppose that QKSD algorithm with the order of $n$, whose the matrix construction method is described in Theorem \ref{theorem:toeplitz_error_matrix} or Theorem \ref{theorem:non_toeplitz_error_matrix}, leads to a generalized eigenvalue problem of $(\tilde{\bm{H}}, \tilde{\bm{S}})$.
Afterward, thresholding algorithm (Algorithm \ref{alg:thresholding}) is applied to $(\tilde{\bm{H}}, \tilde{\bm{S}})$ with the parameter in Eq.\eqref{eq:opt_epsilon} and produces another matrix pair of size ${n_\epsilon}\times {n_\epsilon}$, denoted as $(\tilde{\bm{A}}, \tilde{\bm{B}})$, with the lowest eigenvalue being denoted as $\tilde{E}_0^{(n\rightarrow {n_\epsilon})}$.
In addition to small gap, small perturbation, and random matrix norm bound assumptions, as in Eqs.\eqref{eq:assume_err_bound}, \eqref{eq:assume_gap_condition}, and \eqref{eq:mat_norm_bound}, respectively, it is assumed that the matrix perturbations before ($\eta$) and after ($\chi$) in the thresholding satisfy
\begin{equation}
\label{eq:thm_chi_assumption}
\chi \le \eta.
\end{equation}
Accordingly, the eigenvalue perturbation ($|\tilde{E}^{(n\rightarrow {n_\epsilon})}_0 - E^{(n\rightarrow {n_\epsilon})}_0|$) can be bounded as below, if the relative perturbation is sufficiently small ($|\tilde{E}^{(n\rightarrow {n_\epsilon})}_0 - E^{(n\rightarrow {n_\epsilon})}_0|/E \ll 1$)
\begin{multline}
\label{eq:gep_theorem_result_sampling_noise}
|\tilde{E}^{(n\rightarrow {n_\epsilon})}_{0} - E^{(n\rightarrow {n_\epsilon})}_{0}| \le (1+E^{(n\rightarrow {n_\epsilon})2}_{0})\\ \times \sin^{-1}\left( \frac{\sqrt{2}{n_\epsilon}(e_{\bm{H}}+e_{\bm{S}})}{d_{0}\sqrt{M}} \right),
\end{multline}  
where $d_0^{-1}$ is a condition number for the lowest solution and
\begin{equation*}
\begin{split}
e_{\bm{H}}:=&
\begin{cases}
2\sqrt{2}\|\hat{H}\|_\beta n\sqrt{\log(2n)} \quad &\text{Toeplitz case}\\
2\|\hat{H}\|_\beta n^{3/2}\sqrt{\log(2n)} \quad &\text{Non-Toeplitz case}
\end{cases},\\
e_{\bm{S}}:=& 2\sqrt{2}n\sqrt{\log(2n)}
\end{split}
\end{equation*}
and $M=M_{\bm{H}}+M_{\bm{S}}$ is the total number of samplings optimally distributed, as shown below:
\begin{align}
M_{\bm{H}} = \frac{e_{\bm{H}}}{e_{\bm{H}}+e_{\bm{S}}}M,\label{eq:opt_MH}\\
M_{\bm{S}} = \frac{e_{\bm{S}}}{e_{\bm{H}}+e_{\bm{S}}}M\label{eq:opt_MS}.
\end{align}
\end{theorem}

\begin{proof}
The optimization of a bound for $\eta$, $\mathcal{E}_2:=\sqrt{e_{\bm{H}}^2/M_{\bm{H}}+e_{\bm{S}}^2/M_{\bm{S}}}$, based on the assumption of Eq.\eqref{eq:mat_norm_bound} is obtained by setting Eqs.\eqref{eq:opt_MH} and \eqref{eq:opt_MS} under the constraint of $M=M_{\bm{H}}+M_{\bm{S}}$.
The corresponding optimal value is
\begin{equation}
\mathcal{E}_2^{\star}=\frac{e_{\bm{H}}+e_{\bm{S}}}{\sqrt{M}} \ge \eta \ge \chi.
\end{equation}  
By replacing $\chi$ with $\mathcal{E}_2^{\star}$, we can convert the total perturbation bound in Eq.\eqref{eq:gep_theorem_result} to Eq.\eqref{eq:gep_theorem_result_sampling_noise} through linear approximations:
\begin{multline}
\label{eq:linear_apprx_pert}
\left|\tilde{E} - E\right| = (1+E^2)\left|\tan^{-1}\tilde{E} - \tan^{-1}E\right|\\ + O\left(\left|\tilde{E} - E\right|^2/E\right).
\end{multline}
\end{proof}

Overall, because the number of remaining bases is less than $n$, Theorem \ref{theorem:main_result} leads to an asymptotic bound for finite sampling perturbation,
\begin{equation} \label{eq:pert_bound_bigO}
|\tilde{E}_{0}^{(n\rightarrow {n_\epsilon})} - {E}_{0}^{(n\rightarrow {n_\epsilon})}| = \tilde{O}\left(\frac{\|\hat{H}\|^3 n^z }{d_0 \sqrt{M}}\right),
\end{equation}
where $z=2$ for Toeplitz construction, and $z=5/2$ for non-Toeplitz construction. However, in the numerical simulation, the condition number $d_0^{-1}$ remained large and amplified the sampling error, even though the thresholding alleviated the ill-conditioning.
The following upper bound for $d_0^{-1}$ is found, based on the condition that the least singular value of $\bm{B}$ is larger than $\epsilon$:
\begin{equation} \label{eq:cond_bound}
\begin{split}
d_0^{-1}=&|\bm{x}_0^{\dagger}\bm{B}\bm{x}_0|^{-1}\left(\frac{|\bm{x}_0^{\dagger}\bm{A}\bm{x}_0|^2}{|\bm{x}_0^{\dagger}\bm{B}\bm{x}_0|^2}+1\right)^{-1/2}\\
\le& \epsilon^{-1} \left(E_0^{(n\rightarrow n_{\epsilon})2}+1\right)^{-1/2}.
\end{split}
\end{equation}
Correspondingly, the eigenvalue perturbation in Eq.\eqref{eq:pert_bound_bigO} can be bounded as
\begin{equation}\label{eq:assymptotic_pert_bound_afterd0}
|\tilde{E}_{0}^{(n\rightarrow {n_\epsilon})} - {E}_{0}^{(n\rightarrow {n_\epsilon})}| \le \tilde{O}\left(\frac{\|\hat{H}\| n^z }{\epsilon \sqrt{M}}\right).
\end{equation}
Note that although the perturbation bound decreases as we set larger $\epsilon$, this also causes more basis vectors to be thresholded, resulting in an increase in the projection error, $|E_0^{(n \rightarrow n_\epsilon)} - E_0|$.
Therefore, we can find a trade-off relation between the perturbation effect and the projection error.
The optimal $\epsilon$ is heuristically determined in Eq.\eqref{eq:opt_epsilon} and numerically verified in Section \ref{sec:numerical_truncation}.

When it comes to the condition number, we have utilized $d^{-1}_0$ throughout this work, as it is more suitable for solving the GEVP. 
In addition to $d^{-1}_0$, a different condition number has been adopted in another work \cite[Appendix A]{quantum_power_method}:
\begin{equation}\label{eq:cond_number_s}
\mathrm{cond}(\bm{S}) = \|\bm{S}\|\|\bm{S}^{-1}\|.
\end{equation}
Although this quantity depicts the relative numerical stability to the perturbation with linear equation \cite[Theorem~12.2]{trefethen97}, it also can be related to relative GEVP perturbation because the calculation of $\bm{S}^{-1}$ is involved in GEVP solver and $\mathrm{cond}(\bm{S})$ amplifies the overlap matrix perturbation, $\|\bm{\Delta_S}\|$.
The detailed relation is shown in Appendix \ref{sec:Simpler_GEVP_Perturbation}.
The condition number of $\bm{S}$ becomes large because of this linear dependency on the basis \cite[Appendix A]{quantum_power_method}.
The linear dependency increases especially when the time step is small ($\Delta_t \ll O(\|\hat{H}\|^{-1})$), which makes the base operator $e^{-i\hat{H}\Delta_t}=\hat{I}-i\hat{H}\Delta_t+O(\|\hat{H}\|^2\Delta_t^2)$ close to identity.

\subsection{Trotter and Gate Errors}
Up to this point, we have focused on sampling error analysis. However, in the early fault-tolerant regime, one cannot completely circumvent other sources of noise. 
Here, we present a brief analysis of the matrix perturbation caused by Trotter and gate errors.

When constructing the matrix $\bm{H}$, either Toeplitz or non-Toeplitz construction should be employed.
As mentioned in Section \ref{sec:QKSD}, with a sufficiently small Trotter error, the Toeplitz construction is superior to the non-Toeplitz one because measurements are allocated to obtain fewer elements.
Using Theorems \ref{theorem:toeplitz_error_matrix} and \ref{theorem:non_toeplitz_error_matrix}, we derive the asymptotic bound for the circuit depth, which can suppress the simulation error to make the Toeplitz construction more advantageous than non-Toeplitz one.
The required depth for such circuit is 
\begin{equation}\label{eq:min_ckt_depth}
D^{(\mathrm{t})} = \tilde{\Omega}\left( N_{\Gamma}\Delta_t^2 n^{3/2} M_{\bm{H}}^{1/2} \right),
\end{equation}
whose details are provided in Appendix \ref{sec:appendix_consideration_of_simulation_error}.
Here, note that $N_\Gamma$ is the number of Hamiltonian fragments for which the exact time evolution operators are known and efficiently implementable.
Such deep quantum circuits can reduce the effect of the Trotter error in the Toeplitz construction lower than sampling error in the non-Toeplitz one which does not suffer from the Trotter error.
If the circuit depth of Eq.\eqref{eq:min_ckt_depth} is not satisfied while Toeplitz construction is adopted, the GEVP perturbation caused by the sampling error and Trotterization is bounded as 
\begin{equation}\label{eq:pert_bound_st}
\begin{split}
|\tilde{E}_{0,\mathrm{ST}}^{(n\rightarrow n_\epsilon)}& - {E}_{0,\mathrm{ST}}^{(n\rightarrow n_\epsilon)}| \le\\
\tilde{O}&\left(\frac{\|\hat{H}\|^3 n^2}{d_0}\left(\frac{N_\Gamma \Delta_t^2 n}{D} + \frac{1}{\sqrt{M}}\right)\right).
\end{split}
\end{equation}
Here, $E_{0,\mathrm{ST}}^{(n\rightarrow n_\epsilon)}$ denotes the lowest eigenvalue in the approximated subspace, $\tilde{\mathcal{K}}_n:=\{\hat{U}_{\mathrm{ST}}(k\Delta_t)\ket{\phi_0}\}_{k=0}^{n-1}$, and $\tilde{E}_{0,\mathrm{ST}}^{(n\rightarrow n_\epsilon)}$ is the eigenvalue perturbed by the sampling and Trotter errors, where the thresholding with the parameter $\epsilon$ is employed.
The details of Trotterization perturbation analysis are provided in Appendix \ref{sec:appendix_consideration_of_simulation_error}.

Now, let us focus on the hardware noise.
In early fault tolerant quantum regime, although it includes quantum error correction, a finite level of logical qubit error is induced.
By adopting the Pauli error model introduced in \cite{liang2023modeling}, it can be shown that the hardware noise induces a decay of the measured matrix:
\begin{equation}
\tilde{\bm{Z}}_{\mathrm{H}} \approx e^{-\lambda}\bm{Z}+\bm{\Delta_Z},
\end{equation}
where $\lambda= N_q D \ln r^{-1}$ is determined by single qubit fidelity $r$, the number of qubits $N_q$, and the circuit depth $D$.
Furthermore $\bm{\Delta_Z}$ is an additive error, which behaves similarly to that without hardware noise.
Consequently, the GEVP perturbed from Eq.\eqref{eq:gen_eigeq} is determined as
\begin{equation}\label{eq:hardware_noise_gevp}
(\bm{H}+e^{\lambda}\bm{\Delta_H})\tilde{\bm{c}} = (\bm{S}+e^{\lambda}\bm{\Delta_S})\tilde{\bm{c}}\tilde{E}^{(n)},
\end{equation}
which effectively amplifies the perturbation bound Eqs.\eqref{eq:pert_bound_st} or \eqref{eq:pert_bound_bigO} by the factor of $e^{\lambda}$.
The detailed procedures are presented in Appendix \ref{sec:appendix_Hardware_noise}.

\section{Numerical Simulations}\label{sec:result}
A numerical simulation of many-body systems is performed using the perturbed QKSD method, which considers the error caused by finite sampling of the Hadamard test.
By comparing this simulation with the theoretical results in Section \ref{sec:main_result}, we then analyze the statistical behavior of the error matrices, the validation of the optimal thresholding, and the perturbed ground energy in terms of the finite sampling error.

In this work, we focus on the 1D spinful Fermi-Hubbard model of length $L$, which is given as
\begin{align}
\nonumber \hat{H} = -t\sum_{i=1\cdots L-1}\sum_{\sigma \in \{\uparrow, \downarrow\}}(&\hat{a}^\dagger_{i\sigma}\hat{a}_{{i+1}\sigma} + \hat{a}_{i\sigma}\hat{a}^\dagger_{{i+1}\sigma})\\
&+ u\sum_{i=1\cdots L}\hat{a}^\dagger_{i\uparrow}\hat{a}_{{i}\uparrow}\hat{a}^\dagger_{i\downarrow}\hat{a}_{{i}\downarrow},\label{eq:Hamiltonian}
\end{align}
where $\hat{a}^{(\dagger)}_{i\sigma}$ is a fermionic annihilation (creation) operator at the $i$\textsuperscript{th} site with spin $\sigma\in\{\uparrow,\downarrow\}$.
First, three parameter cases are considered: $(t, u)=(0.1, 0.2), (0,2, 0.1),$ and $(0.1, 0.8)$ cases, which correspond to stronger coupling as $u/t$ increases \cite{annurev:/content/journals/10.1146/annurev-conmatphys-090921-033948}.
Although the parameter ratio of $t$ to $u$ matters, the absolute values are set making  $\|\hat{H}\|_\beta$ close to unity.
We particularly select the system size $L=8$ in this simulation.
The Hamiltonian for each setting is mapped in terms of Pauli operators using a qubit mapping scheme that embraces the number, spin, and spatial symmetries, which results in reducing three qubits.
Thus, the numbers of remaining qubits is $N_q=13$.

\begin{figure}[t!]
	\centering
	\includegraphics[width=\linewidth]{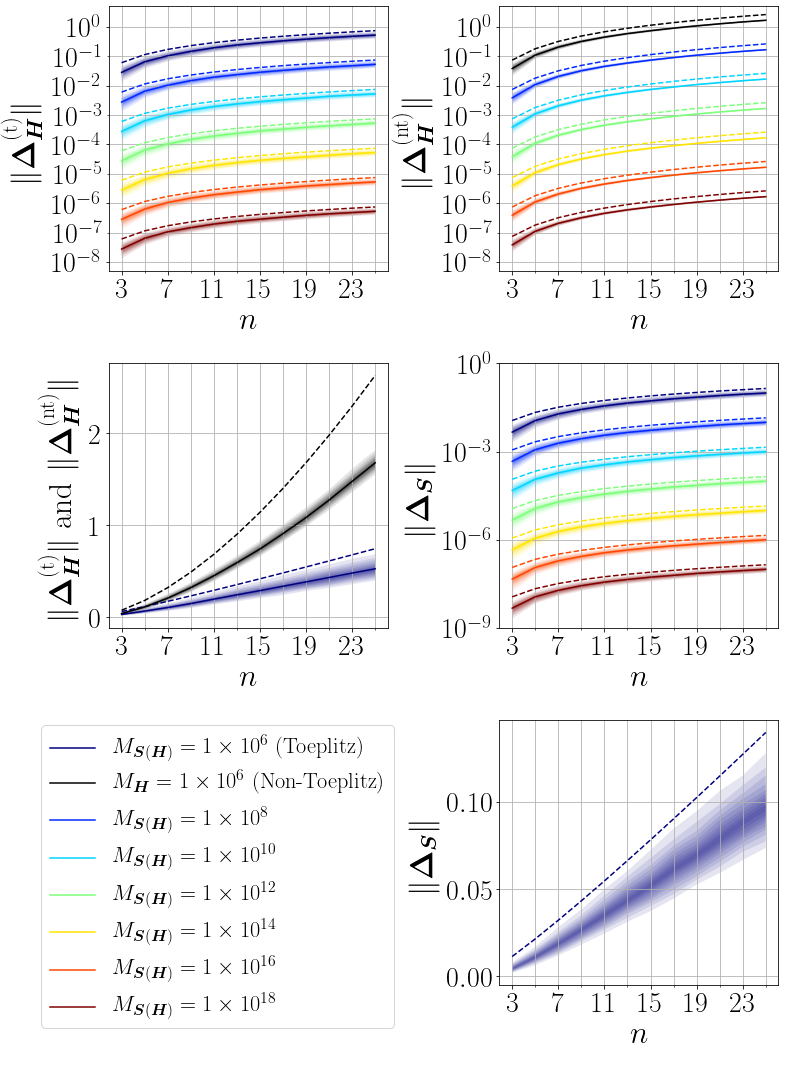}
	\caption{Distribution and estimated bound for $\|\bm{\Delta_H}\|$ and $\|\bm{\Delta_S}\|$, over 10,000 different random matrices per point along the different Krylov orders $n$ (parameters in Eq.\eqref{eq:Hamiltonian}: $L=8$ and $u/t=8$).
		For the case of $\bm{\Delta_H}$, Toeplitz($\bm{\Delta}_{H}^{(\mathrm{t})}$) and non-Toeplitz($\bm{\Delta}_{H}^{(\mathrm{nt})}$) construction methods are considered, whereas only the Toeplitz construction is used for $\bm{\Delta_S}$.
		Each color in the plots denotes different $M_{\bm{H}}$ and $M_{\bm{S}}$ values, and the dashed line represents the upper bound in Eqs.\eqref{eq:toeplitz_error_matrix_norm_bound} and \eqref{eq:non_toeplitz_error_matrix_norm_bound}, where the solid lines and shadows represent the ensemble average and histogram populations, respectively.
		This figure also shows the linearly scaled views for $M_{\bm{H}(\bm{S})}=10^{6}$, which are used to emphasize the precise growth relationship between Toeplitz ($\sim n\sqrt{\log n}$) and the non-Toeplitz ($\sim n^{3/2}\sqrt{\log n}$) for the construction of $\bm{H}$.}
	\label{fig:error_norm}
\end{figure}

Unitary decomposition of the Hamiltonian in Eq.\eqref{eq:pauli_ham_grp} is performed using the sorted insertion grouping algorithm \cite{efficient_quantum_measurement}, which results in $N_\beta = 34$.
It is assumed that we perfectly simulate the time evolution $e^{-i\hat{H}t}$, thereby neglecting Suzuki-Trotterization errors and hardware decoherence, in order to focus on the sampling error.
Furthermore, the Krylov basis of order $n$ is set to $\{\exp{(-ik\Delta_t\hat{H})}\ket{\phi_0}\}_{k=-\lfloor n/2 \rfloor}^{\lfloor n/2\rfloor}$, where $\ket{\phi_0}$ is the Hartree Fock ground state.
Here, we choose time steps as the time sequence from Theorem \ref{theorem:qksd_err}.
Since it is difficult to determine $\Delta E_{N-1}$ in advance, we replace it with $\|\hat{H}\|_\beta$.

\subsection{Sampling Error Matrices}

Statistical results for the error matrix norms ($\|\bm{\Delta_H}^{(\mathrm{t})}\|$, $\|\bm{\Delta_H}^{(\mathrm{nt})}\|$, and $\|\bm{\Delta_S}\|$) generated by the 10,000 different random seeds are shown in Fig.\ref{fig:error_norm}, together with their probabilistic bound, Eqs.~\eqref{eq:toeplitz_error_matrix_norm_bound} and \eqref{eq:non_toeplitz_error_matrix_norm_bound}. 
We can clearly see that the error norm is inversely proportional to the square root of the sampling number. 
Although Fig.\ref{fig:error_norm} displays only the case of $L=8$ and $u/t=8$, the tendency is consistent for the various parameters in the Hamiltonian.
In Fig.\ref{fig:error_norm}, with Krylov order $n$, the error norm is bounded by $O(n^{3/2}\sqrt{\log n})$ for the non-Toeplitz construction and $O(n \sqrt{\log n})$ for the Toeplitz construction.
As Krylov order $n$ and number of samplings $M$ increase, the logarithmic difference between the mean and the bound remains mostly constant.
In other words, the factor $\kappa$ in Eq.\eqref{eq:mat_norm_concentration} is invariant under the fixed target Hamiltonian and matrix construction method.
However, although the tail bound probability can be estimated using Eq.\eqref{eq:mat_norm_concentration}, all experiments in this study were proved to be within the estimated bound, validating our assumption in Eq.\eqref{eq:mat_norm_bound}.

\subsection{Truncation of Unstable Basis}\label{sec:numerical_truncation}

\begin{figure}[!ht]
	\centering
	\includegraphics[width=\linewidth]{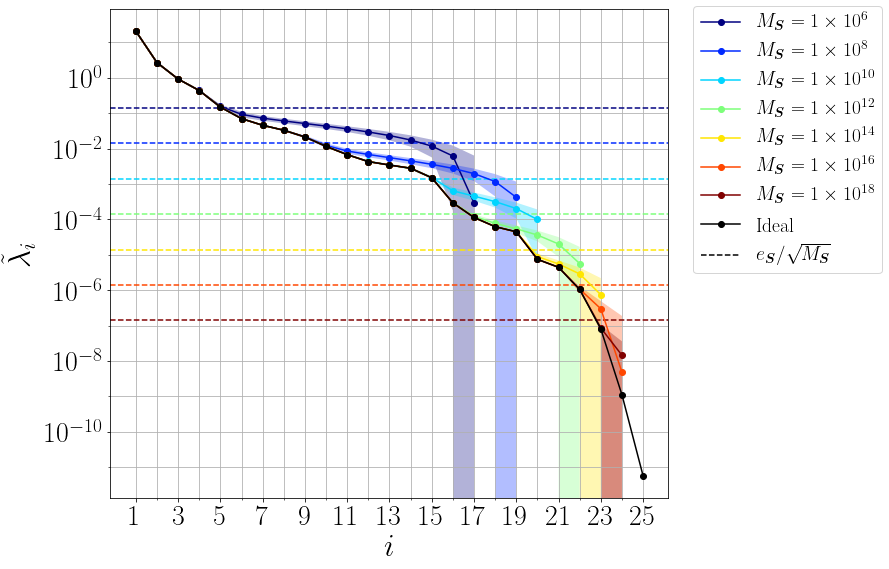}
	\caption{
		Singular values of $\bm{S}$, $\lambda_i$(black), and those of its sampling perturbations $\tilde{\bm{S}}$, $\tilde{\lambda}_m$(other colors) under the Krylov order of $n=25$ (Hamiltonian setting: $L=8$ and $(t, u)=(0.2, 0.1)$).
		They are in nonincreasing order and indexed by $i\in\{1\cdots 25\}$.
		Each color index indicates the number of samplings ($M_{\bm{S}}$) to construct $\tilde{\bm{S}}$.
		The solid lines describe the average ordered singular values at different settings, and the colored regions cover the standard deviations from the averages.
		The horizontal dashed lines indicate the bound for the error matrix norm $e^{(\mathrm{t})}_{\bm{S}}/\sqrt{M_{\bm{S}}}$, which is also the heuristic optimal $\epsilon$ in Eq.\eqref{eq:opt_epsilon}.
		Points for the averaged singular values less than zero are not shown for the simplicity.
	}
	\label{fig:S_eigenvalue_perturbation}
\end{figure}
We analyze the effect of thresholding to reduce errors in the perturbed solution and assess the optimal threshold, $\epsilon$, based on the knowledge of the error matrix norm discussed in the previous section.
Fig.\ref{fig:S_eigenvalue_perturbation} depicts the singular-value spectrum of $\tilde{\bm{S}}$ to obtain an insight into selecting the optimal value of $\epsilon$.
This result confirms that, up to the level of $e_{\bm{S}}^{(\mathrm{t})}/\sqrt{M_{\bm{S}}}$, the singular values of $\tilde{\bm{S}}$, referred to as $\tilde{\lambda}_i$ for $i\in\{1,\cdots, n\}$, coincide with the ideal ones, $\lambda_i$ (represented by the dashed line in Fig.\ref{fig:S_eigenvalue_perturbation}).
Below this level, however, a discrepancy between $\tilde{\lambda}_i$ and $\lambda_i$ occurs on a logarithmic scale, which is attributed to Weyl's inequality, stating that an error matrix norm can limit the singular value perturbation.
This finding agrees with the assumption in Eq.\eqref{eq:mat_norm_bound}, which indicates that it is possible to deduce the threshold $\epsilon$ regardless of prior experiments, 
\begin{align}\label{ineq:trunc_pt}
|\tilde{\lambda}_i - \lambda_i| \le \|\bm{\Delta_S}\| < \frac{e^{(\mathrm{t})}_{\bm{S}}(n)}{\sqrt{M_{\bm{S}}}}=\epsilon.
\end{align}

In order to verify the optimality of such thresholding, we investigate the eigenenergy error for different thresholding points.
Fig.\ref{fig:sol_pert} illustrates the relative difference between the perturbed Krylov and exact ground energies as a function of $n^{(\mathrm{t})}$, the dimensions of the remaining basis after thresholding.
As shown in Fig.\ref{fig:sol_pert}, the error of the perturbed Krylov solution $|\tilde{E}^{(n\rightarrow n^{(\mathrm{t})})}_0 - E_0|$ starts to diverge at a certain point because of ill-conditioning, whereas the ideal Krylov solution error, $|E^{(n\rightarrow n^{(\mathrm{t})})}_0 - E_0|$ (the black dashed line), consistently decreases as the number of effective bases $n^{(\mathrm{t})}$ increases.
Furthermore, we validate that if one uses more samplings $M$, the divergence of the solution error appears at a larger $n^{(\mathrm{t})}$.
This divergent point can be estimated closely using the threshold in Eq.\eqref{eq:opt_epsilon}, as represented by the larger marker in Fig.\ref{fig:sol_pert}.
Note that while such points serve as the optimal thresholds for most Toeplitz construction cases, the non-Toeplitz construction has tolerable corresponding energy differences, even though some threshold points deviate by one or two from the optimal values in Fig.\ref{fig:sol_pert}.

\begin{figure}[!t]
	\centering
	\includegraphics[width=\linewidth]{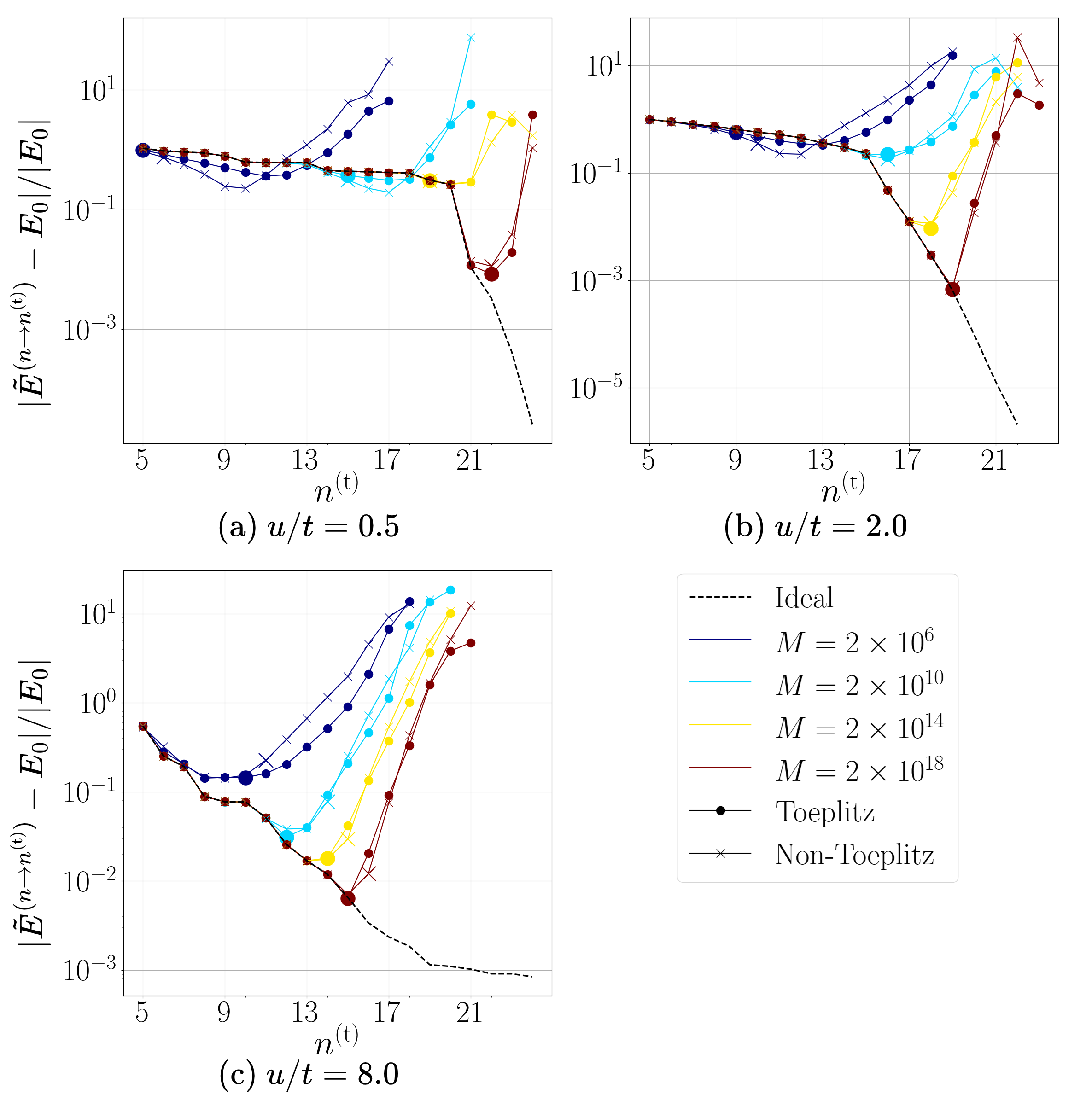}
	\caption{
		Root mean square difference between the Krylov solution and the true ground state energy. It is obtained from the direct diagonalization in the full Hilbert space with different numbers of remaining basis, $n^{(\mathrm{t})}$ for a given Krylov order $n=25$.
		Each subplot shows the result for three different parameters ($u/t$) in Eq.\eqref{eq:Hamiltonian} with $L=8$.
		The black dashed lines correspond to the ideal (unperturbed) Krylov algorithm, with solid lines of other colors indicating experiments with different sampling settings and distributions of Eqs.\eqref{eq:opt_MH} and \eqref{eq:opt_MS}.
		The two marker shapes distinguish the matrix construction methods for $\bm{H}$.
		The threshold points of $n^{(\mathrm{t})}=n_\epsilon$ with the parameter $\epsilon=e_{\bm{S}}^{(\mathrm{t})}/\sqrt{M_{\bm{S}}}$ are emphasized with larger markers.
	}
	\label{fig:sol_pert}
\end{figure}

\subsection{Final Perturbations and Their Bounds}\label{sec:final_pert_bd}
\begin{figure}[t]
	\centering
	\includegraphics[width=\linewidth]{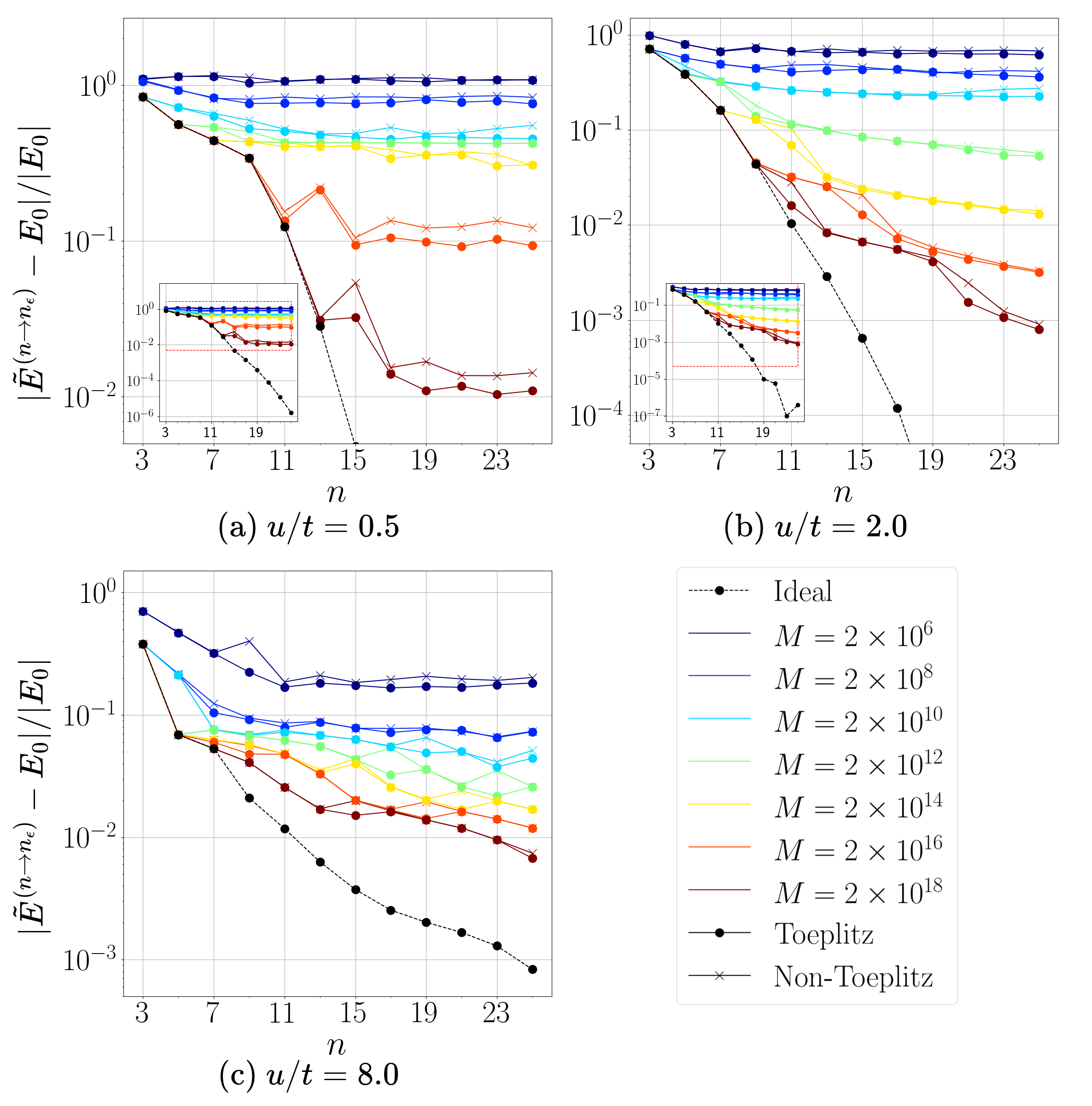}	
	\caption{
		Root mean square difference between the Krylov solution after the optimal threshold of $\epsilon=e_{\bm{S}}^{(\mathrm{t})}/\sqrt{M_{\bm{S}}}$ (the collection of large markers in the Fig.\ref{fig:sol_pert} with different $n$) and the ground state energy found in the full Hilbert space.
		The horizontal axis of each figure represents Krylov order, $n$, and the vertical axis indicates the sampling perturbation, $|\tilde{E}^{(n\rightarrow n_\epsilon)}_0-E_0|/|E_0|$.
		The black dashed lines represent the unperturbed Krylov solution without thresholding, whereas colored solid lines are the perturbed solution where the optimal thresholding criterion is applied. Two different marker shapes indicate the matrix construction methods of $\bm{H}$.
	}
	\label{fig:sol_pert_opt_trunc}
\end{figure}
Based on the optimal threshold described in Section \ref{sec:numerical_truncation}, we depicts the error results and the Krylov order $n$ in Fig.\ref{fig:sol_pert_opt_trunc}.
The eigenenergy errors tend to decrease as the Krylov order increases, particularly with a larger sampling number. This is because for large $n$, an effective basis of higher dimension can be utilized, and small sampling errors give rise to fewer basis vectors being truncated out.
In addition, the non-Toeplitz construction results show slightly more errors than the Toeplitz construction results.
This disparity obviously occurs because the non-Toeplitz construction suffers more sampling errors as the number of samplings becomes more distributed.
The error difference between two constructions is based on the fact that the non-Toeplitz construction requires to obtain $O(n^2)$ elements instead of $O(n)$ in the Toeplitz case.
Although Fig.\ref{fig:sol_pert_opt_trunc} indicates that the Toeplitz construction with large $n$ is advantageous, note that this result does not included the Trotter error which degrades the performance of Toeplitz constructions and increases as simulation time $n\Delta_t$ becomes larger.
In practice, to mitigate the Trotter error sufficiently small and take advantage of Toeplitz construction, the circuit depth represented in Eq.\eqref{eq:min_ckt_depth} is required.

 \begin{figure}[t]
	\centering
	\includegraphics[width=\linewidth]{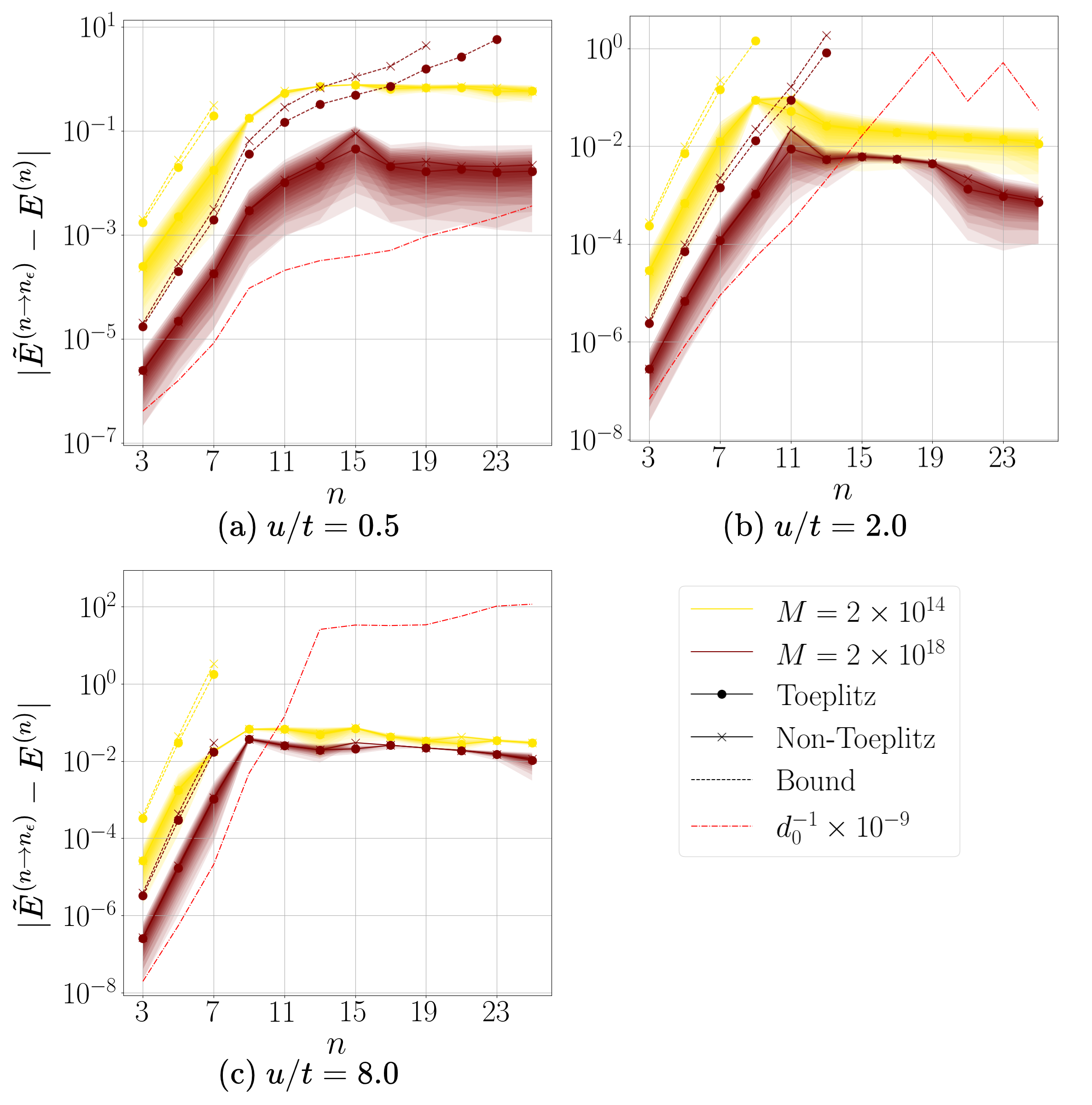}
	\caption{
		Distributions and averages of the sampling perturbations $|\tilde{E}^{(n\rightarrow n_\epsilon)}-E^{(n)}|$ with the optimal thresholding $\epsilon=e^{(\mathrm{t})}_{\bm{S}}/\sqrt{M_{\bm{S}}}$, the perturbation bounds in Eq.\eqref{eq:gep_theorem_result_sampling_noise} and condition number $d_0^{-1}$ are shown for various Hamiltonian settings along with the Krylov order $n$.
		Only two sampling number settings ($M=2\times10^{14}$ and $2\times10^{18}$) are plotted for simplicity.
		The two different marker shapes indicate the construction method for $\bm{H}$.
	}
	\label{fig:sampling_perturbation}
\end{figure}

Fig.\ref{fig:sampling_perturbation} shows the ensemble results, indicating the effects of sampling and the bounds described in Eq.\eqref{eq:gep_theorem_result_sampling_noise}.
The perturbations follow the tendency of the condition number $d_0^{-1}$, plotted together with the shifted log scale.
After the condition number initially increases exponentially with the Krylov basis, it decreases or becomes saturated at a certain level.
Although the condition number largely affects the perturbations, obtaining analytical and prior $d_0^{-1}$ information remains challenging.
In addition, the speculation about sampling perturbations (Theorem \ref{theorem:main_result}) precisely estimates the upper bound for the effect of finite sampling unless the condition number is so large that the $\sin^{-1}$ in Eq.\eqref{eq:gep_theorem_result_sampling_noise} cannot be evaluated.

\section{Conclusion}\label{sec:conclusion}
Based on random matrix theory, we have developed a theoretical framework to analyze the validity of QKSD regarding sampling noise using a non-asymptotic bound.
The analysis involves an accurate modeling of sampling errors and an optimal shot distribution for each configuration, minimizing the overall sampling variance.
We also derived a theorem that precisely estimates finite sampling errors when the condition numbers are small, which was verified by numerical simulations of the 1D Fermi-Hubbard model.
Finally, finite sampling generates eigenenergy noises by $O(n^z \sqrt{\log n} /(d_0 \sqrt{M}))$, where $z=2$ for the Toeplitz construction, and $z=5/2$ for the non-Toeplitz construction.
This is because shots are distributed to obtain more independent elements ($O(n^2)$) in non-Toeplitz than the Toeplitz case ($O(n)$).
The choice of construction method depends on the availability of an accurate simulation of $e^{-i\hat{H}t}$ with an error smaller than the sampling error.
If such a simulation is unavailable, one can use a non-Toeplitz basis instead, which exacerbates the sampling complexities, to avoid simulation errors.

Furthermore, we presented a near-optimal criteria for basis thresholding to manage large condition numbers by quantifying the sampling errors in the overlap matrix.
This criteria is applied to determine the stable basis of the overlap matrix when solving the GEVP in the QKSD method.
The numerical experiments have showed that the error is mostly minimized with the threshold that we presented. 
Unfortunately, the condition number tended to increase exponentially with the Krylov order, which leads to amplify the noise significantly, even when basis thresholding was applied. 
Because some noises, such as finite sampling and the Trotter error, are often related to time complexity, the condition number determines the slope of the trade-off between the runtime and the accuracy of the algorithm.
We found an upper bound for the condition number, which is inversely proportional to the threshold.

Although several researchers have utilized other strategies to mitigate this numerical instability, fully overcoming this problem remains a challenge.
The first strategy is that, if one uses orthogonalization, such as the Gram-Schmidt process, the condition number of $\bm{S}$ ($\mathrm{cond}(\bm{S}) =\|\bm{S}\|\|\bm{S}^{-1}\|$) becomes unity, making the problem restore the numerical stability.
However, stating a new orthonormal basis as a linear combination of the Krylov basis involves the calculation of $\tilde{\bm{\Lambda}}^{-1/2}$, where $\tilde{\bm{\Lambda}}$ is a diagonal matrix whose elements are singular values of $\tilde{\bm{S}}$.
This situation is the same as generalized eigenvalue perturbation; thus, orthogonalization suffers from the problem of the condition number as well.
Another approach, the Lanczos method \cite[Appendix F]{quantum_power_method}, utilizes an orthonormal basis to make $\bm{H}$ tridiagonal.
Although the eigenvalue perturbation is small, this approach depends on the Hamiltonian moments, $\braket{\hat{H}^k}$, whose efficient calculation method with a tolerable additive error is still unknown.
In addition to approaches based on orthogonal bases, different forms of the Krylov basis can be used to construct a numerically stable overlap matrix $\bm{S}$ whose column vectors are linearly independent.
As an example of a different Krylov basis, the QPM highlights that the linear dependency among the bases generated by the Hamiltonian power tends to be lower than that of the time evolution basis \cite[Appendix~A]{quantum_power_method}.
However, this basis is susceptible to noise amplification when it is realized in the quantum circuit, which can pose practical challenges when the norm of the Hamiltonian is large. 

Many recent quantum-classical hybrid algorithms incorporate the estimation of overlaps among states to construct matrices and utilize classical methods to solve GEVP or linear equation \cite{PhysRevResearch.3.043212, PhysRevA.104.042418}.
If the computation of inverse matrices is involved in the classical method, it is highly probable that no one can avoid the finite sampling error in coping with the large condition number. 
We thus expect our analysis to provide great impetus to other similar works that require a meticulous analysis of errors stemming from finite sampling.


\begin{acknowledgments}
This work was supported by Basic Science Research Program through the National Research Foundation of Korea (NRF), funded by the Ministry of Education, Science and Technology (NRF-2022M3H3A106307411, NRF-2022M3K2A1083890, NRF-2023M3K5A1094805, NRF-2023M3K5A1094813) and Institute for Information \& communications Technology Promotion (IITP) grant funded by the Korea government(MSIP) (No. 2019-0-00003, Research and Development of Core technologies for Programming, Running, Implementing and Validating of Fault-Tolerant Quantum Computing System).
\end{acknowledgments}

\bibliographystyle{quantum}
\bibliography{quantum-template}

\onecolumn
\appendix

\section{Anticommutation (Unitary) Grouping}
\label{sec:appendix_hamiltonian_partitioning}
The qubit Hamiltonian $\hat{H}$ is written as
\begin{equation}
\label{eq:pauli_ham}
\hat{H} = \sum_{k=1}^{N_{P}} \alpha_k \hat{P}_k,
\end{equation}
where $\hat{P}_k\in\{\hat{I}, \hat{\sigma_x}, \hat{\sigma_y}, \hat{\sigma_z}\}^{\otimes N_{q}}$ is a Pauli string acting on $N_q$ qubits, and $\alpha_k\in\mathbb{R}$ is the corresponding coefficient.
Because the number of Pauli strings, $N_P$, considered in the Hamiltonian is often considerably less than its maximum possibilities, $4^{N_q}$, the time propagation unitary can be efficiently realized.

The unitary partitioning was proposed to construct an anticommutation grouping that represents the Hamiltonian as a sum of positively weighted unitaries \cite{unitary_partitioning, LCU_Love}:
\begin{equation}
\hat{H} = \sum_{j=1}^{N_\beta} \beta_j \hat{U}_j.
\end{equation}

In such partitioning, each of the $N_\beta$ partitions, $\hat{U}_j$, that contain Pauli strings of a set $\mathcal{S}(j)$,
\begin{equation}
\hat{U}_j = \sum_{k\in\mathcal{S}(j)} \alpha^\prime_k \hat{P}_k,
\end{equation}
becomes unitary if the following conditions hold:
\begin{alignat*}{3}
\mathrm{Im} (\alpha^{\prime*}_k \alpha^\prime_l)&=0  \qquad &\forall k, l &\in \mathcal{S}(j),\\
\{\hat{P}_k, \hat{P}_l\}&=2\delta_{k,l} \hat{I}^{\otimes N_q} \qquad &\forall k, l &\in \mathcal{S}(j),\\
\sum_{k\in\mathcal{S}(j)}\left|\alpha^{\prime}_k\right|^2 &= 1.
\end{alignat*}

Here, the positive coefficients, $\beta_j=(\sum_{k\in\mathcal{S}(j)}\alpha_k^2)^{1/2}$, recover the normalization of $\alpha^{\prime}_k$ to $\alpha_k$.
The partitioning strategy to minimize the number of partitions, $N_\beta$, is equivalent to the minimum clique cover problem ($\texttt{MCC}$) of a graph whose nodes of Pauli words $\hat{P}_k$ are connected to each other if two Pauli words anticommute.
Not only is the $\texttt{MCC}$ an NP-hard problem, but it is also not equivalent to minimizing the sampling error, as described in Appendix \ref{sec:hamiltonian_overlap}.
Instead, it can be heuristically minimized with the $\texttt{SORTED INSERTION}$ algorithm \cite{efficient_quantum_measurement}, which is employed as the default technique throughout this study to determine the partitioning.

Applying unitary partitioning with the Hadamard test, for example, $\bra{\phi_1}\hat{H}\ket{\phi_2} = \sum_j\beta_j\bra{\phi_1}\hat{U}_j\ket{\phi_2}$ for some states $\ket{\phi_1}$ and $\ket{\phi_2}$, we then estimate the overlaps for Krylov matrix elements.
Any unitary grouping other than individual Pauli grouping reduces the sampling variances from $O(\|\hat{H}\|_\alpha^2/M)$ to $O(\|\hat{H}\|_\beta^2/M)$, with $\|\hat{H}\|_\alpha:=\sum_k{|\alpha_k|} \ge \|\hat{H}\|_\beta:=\sum_j{\beta_j}$ and the sampling number $M$.

Although one can use other algorithms for overlap estimation, such as the extended swap test \cite{quantum_filter_method}, the process and results of the analysis in this study can be applied in a similar way. However the anticommutation grouping is replaced by commutation grouping, which corresponds to solving the clique covering problem of the complement graph of the anticommutation graph.

\section{Generalized Eigenvalue Perturbations With Thresholding}
\label{sec:appendix_generalized_eignevalue_perturbation}

Here, we review another version of Theorem \ref{theorem:generalized_eigenvalue_perturbation} using previous results \cite[Theorem~2.7]{Theory_QSDK} considering thresholding.

\begin{theorem}[Generalized Eigenvalue Perturbations with thresholding, {\cite[Theorem~2.7]{Theory_QSDK}}]
	For perturbed pair $(\tilde{\bm{H}}, \tilde{\bm{S}})=(\bm{H}+\bm{\Delta_H}, \bm{S}+\bm{\Delta_S})$ with size $n\times n$ and positive definite $\bm{S}$, let $\tilde{\bm{S}}=\tilde{\bm{V}}^{\dagger}\bm{\tilde{\Lambda}_S}\tilde{\bm{V}}$ be diagonalized by $\tilde{\bm{V}}$, where $\bm{\tilde{\Lambda}_S}=\mathrm{diag}\{\tilde{\lambda}_1, \cdots \tilde{\lambda}_{n}\}$ is ordered by nonincreasing $\tilde{\lambda}$'s.
	Based on threshold parameter $\epsilon>0$, the column vectors in $\tilde{\bm{V}}$ are cut off if the corresponding singular values of $\tilde{\lambda}$ are less than $\epsilon$, resulting in a truncated $n\times {n_\epsilon}$ matrix $\bm{\tilde{V}}_{>\epsilon}$, where ${n_\epsilon}$ is the number of remaining column vectors, i.e., $\tilde{\lambda}_{{n_\epsilon}}\ge\epsilon\ge\tilde{\lambda}_{{n_\epsilon}+1}$. Similarly, let $\bm{\Lambda_S}=\mathrm{diag}\{\lambda_1, \cdots \lambda_n\}$ and $\bm{V}_{>\epsilon}$ be the singular values and the truncated basis matrices obtained from the unperturbed pair $(\bm{H}, \bm{S})$.
	Consider another pair $(\bm{A}, \bm{B})=(\bm{V}_{>\epsilon}^{\dagger}\bm{H}\bm{V}_{>\epsilon},\bm{V}_{>\epsilon}^{\dagger}\bm{S}\bm{V}_{>\epsilon})$ as the thresholded pair and its perturbed pair $(\tilde{\bm{A}}, \tilde{\bm{B}})$ obtained from $(\tilde{\bm{H}}, \tilde{\bm{S}})$ with conjugation by $\tilde{\bm{V}}_{>\epsilon}$.
	Determine if the following three assumptions hold:
	\begin{enumerate}[leftmargin=\assumptionmargin,label=\assumptionlabel \arabic* ]
		\item Pair $(\bm{H}, \bm{S})$ satisfies the following geometric mean bound for some parameters $\mu>0$ and $0\le\alpha\le1/2$,
		\begin{equation}
		\left| \bm{v}_i^{\dagger} \bm{H} \bm{v}_j \right| \le \mu \min(\lambda_i, \lambda_j)^{1-\alpha}\max(\lambda_i, \lambda_j)^{\alpha} 
		\end{equation}
		for all $1\le i,j\le n$,
		where $\bm{v}_i$ is the $i^{\mathrm{th}}$ eigenvector of $\bm{S}$.\label{a:1_Theory_QSDK}
		\item The largest singular value of matrix $\bm{S}$ that is smaller than the threshold, denoted as $\lambda_{n_\epsilon}$, and the next singular value, denoted as $\lambda_{{n_\epsilon}+1}$, are sufficiently separated such that the following inequality holds:
		\begin{equation}
		\lambda_{{n_\epsilon}+1}+\|\bm{\Delta_S}\|\le \epsilon <(1+\rho)\epsilon \le \lambda_{{n_\epsilon}}
		\end{equation}
		for some $\rho > 0$.\label{a:2_Theory_QSDK}
		\item Noise $\|\bm{\Delta_S}\|$ is sufficiently small such that $(1+\rho^{-1})\|\bm{\Delta_S}\|/\epsilon\le 1$.\label{a:3_Theory_QSDK}
	\end{enumerate}
	
	If these hold, then the perturbation on $(\bm{A}, \bm{B})$ is bounded as 
	\begin{equation}
	\label{eq:precise_def_chi}
	\chi := \sqrt{\|\bm{W}^\dagger\tilde{\bm{A}}\bm{W} - \bm{A}\|^2 + \|\bm{W}^\dagger\tilde{\bm{B}}\bm{W} - \bm{B}\|^2} \le 3(2+\mu)n^3(1+\rho^{-1})\left(\frac{\|\bm{S}\|}{\epsilon}\right)^\alpha \|\bm{\Delta_S}\| + \|\bm{\Delta_H}\|
	\end{equation}
	for conjugation, $\bm{W}:=\tilde{\bm{V}}_{>\epsilon}^\dagger\bm{V}_{>\epsilon}$.
	
	Furthermore, let $E^{(n)}_0$ and $E^{(n)}_1$ be the ground and first excited eigenvalues of the pair $(\bm{A}, \bm{B})$. Suppose the following assumptions:
	\begin{enumerate}[leftmargin=\assumptionmargin,label=\assumptionlabel \arabic* ]\addtocounter{enumi}{3}
		\item Error bound $\chi$ is sufficiently small: ${n_\epsilon}\chi \le \lambda_{n_\epsilon}.$\label{a:4_Theory_QSDK}
		\item Gap condition $|\tan^{-1}E^{(n)}_1 - \tan^{-1}E^{(n)}_0|\ge \sin^{-1}\frac{{n_\epsilon}\chi}{\lambda_m}$ holds\label{a:5_Theory_QSDK}.
	\end{enumerate}
	Then, with the condition number that corresponds to the eigenangle $\tan^{-1}E^{(n)}_0$, $d_0^{-1}$, the eigenangle recovered by the thresholding applied to $(\tilde{\bm{H}}, \tilde{\bm{S}})$ satisfies
	\begin{equation}
	\left| \tan^{-1}E^{(n)}_0 - \tan^{-1}\tilde{E}^{(n)}_0 \right| \le \sin^{-1}\frac{{n_\epsilon}\chi}{d_0}.
	\end{equation}
	
\end{theorem}

Based on this theorem, it was previously mentioned \cite{Theory_QSDK} that \ref{a:1_Theory_QSDK} necessarily holds if two parameters are chosen $\mu=\max|\Lambda(\bm{H}, \bm{S})| \lessapprox \|\hat{H}\|$ and $\alpha=1/2$.
Although Eq.\eqref{eq:precise_def_chi} differs from the definition in Eq.\eqref{eq:def_chi}, Eq.\eqref{eq:precise_def_chi} can also explain the perturbation for $(\tilde{\bm{A}}, \tilde{\bm{B}})$ because the perturbed pair $(\tilde{\bm{A}}, \tilde{\bm{B}})$ and its conjugated pair, $(\bm{W}^{\dagger}\tilde{\bm{A}}\bm{W}, \bm{W}^{\dagger}\tilde{\bm{B}}\bm{W})$ are identical problems.
Furthermore, \ref{a:4_Theory_QSDK} in \cite{Theory_QSDK} is analogous to
\begin{equation*}
r(\bm{\Delta_A}, \bm{\Delta_B})\|\bm{X}\|^2 < c(\bm{A},\bm{B}),
\end{equation*}
in another study \cite[Corollary~3.6]{Mathias2004TheDG}, where $\bm{X}$ is a matrix that diagonalizes $\bm{A}+i\bm{B}$ and $r(\bm{\Delta_A}, \bm{\Delta_B}):=\max\{|\bm{x}^{\dagger}(\bm{\Delta_A}+i\bm{\Delta_B})\bm{x}|:\bm{x}^{\dagger}\bm{x} = 1\}$ is defined.
Also, the Crawford number $c(\bm{A}, \bm{B})$, can be more tightly bounded by $\lambda_{n_\epsilon}$, rather than $\epsilon$ in \cite{Theory_QSDK}:
\begin{equation*}
c(\bm{A}, \bm{B}) := \min_{\bm{x}\in\mathbb{C}^n:\bm{x}^\dagger\bm{x}=1}{\left\{ \sqrt{|\bm{x}^{\dagger}\bm{A}\bm{x}|^2+ |\bm{x}^{\dagger}\bm{B}\bm{x}|^2 } \right\}}
\ge \left\|\bm{B}^{-1}\right\|^{-1}=\lambda_{n_\epsilon}.     
\end{equation*}

\section{Hadamard Test Variance}
\label{sec:appendix_hadamard_sampling_variance}
Consider the Hadamard tests performed independently for the real and imaginary parts with $m^{(\mathrm{r})}$ and $m^{(\mathrm{i})}$ samplings, respectively.
Moreover, under such sampling setting, let $Z=R+iI$ be an estimator for $\braket{\psi|\hat{U}|\psi}$, where $R$ represents the real part of $Z$, and $I$ represents the imaginary part.
The expectation value of $Z$ is given by,
\begin{equation}\label{eq:hadamard_noiseless_expectation}
\mathbb{E}[Z] = \mathbb{E}[R] + i\mathbb{E}[I] := \mathrm{Re}[\bra{\psi}\hat{U}\ket{\psi}] + i \mathrm{Im}[\bra{\psi}\hat{U}\ket{\psi}].
\end{equation}
If one labels the probability to find $\ket{0}$ at the measurement of ancillary qubit in the single real and imaginary settings as $p$ and $q$, respectively, the expectation values are represented as
\begin{equation}\label{eq:probability_hadamard}
\begin{split}
\mathbb{E}[R] =& 2p - 1, \\
\mathbb{E}[I] =& 2q - 1.
\end{split}
\end{equation}
Furthermore, the random variables $R$ and $I$ are averaged Bernoulli distributions, 
\begin{equation}
\begin{split}
R \sim \frac{2}{m^{(\mathrm{r})}}\mathrm{Bin}(m^{(\mathrm{r})}, p) - 1,\\
I \sim \frac{2}{m^{(\mathrm{i})}}\mathrm{Bin}(m^{(\mathrm{i})}, q) - 1,
\end{split}
\end{equation}
where $\mathrm{Bin}(m, p)$ denotes a binomial distribution with $m$ samplings and a success probability of $p$.

Thus, the sampling variances become
\begin{equation}
\begin{split}
\mathrm{Var}[R] &= \frac{4p(1-p)}{m^{(\mathrm{r})}} = \frac{1-\mathbb{E}[R]^2}{m^{(\mathrm{r})}},\\
\mathrm{Var}[I] &= \frac{4q(1-q)}{m^{(\mathrm{i})}} = \frac{1-\mathbb{E}[I]^2}{m^{(\mathrm{i})}}.
\end{split}
\end{equation}
 The variance of $Z$ is then given as
\begin{equation}\label{eq:hadamard_noiseless_variance}
\begin{split}
\mathrm{Var}[Z] &= \mathrm{Var}[R] + \mathrm{Var}[I] \\ 
&=  \frac{1 - \mathrm{Re}[\bra{\psi}\hat{U}\ket{\psi}]^2}{m^{(\mathrm{r})}} + \frac{1 -\mathrm{Im}[\bra{\psi}\hat{U}\ket{\psi}]^2}{m^{(\mathrm{i})}}.     
\end{split}
\end{equation}
Furthermore, it is minimized under the constraint of shots, $m^{(\mathrm{r})}+m^{(\mathrm{i})}=m$, by assigning 
\begin{equation}\label{eq:hadamard_shot_alloc}
\begin{split}
m^{(\mathrm{r})} &\propto \sqrt{1-\mathrm{Re}[\braket{\psi|\hat{U}|\psi}]^2}\\
m^{(\mathrm{i})} &\propto \sqrt{1-\mathrm{Im}[\braket{\psi|\hat{U}|\psi}]^2}.
\end{split}
\end{equation}

To remove the dependency of $\braket{\psi|\hat{U}|\psi}$, which is unknown in prior, Haar measure averaging is applied both on $\hat{U}$ and another unitary, $\hat{V}$ such that $\ket{\psi} = \hat{V}\ket{0}$ (i.e. Fubini-Study measure averaging on $\ket{\psi}$) as below:
\begin{equation}
\mathbb{E}_{\hat{U},\hat{V}\sim \mathcal{H}(\mathcal{U}_N)}\left[\left|\bra{0}\hat{V}^{\dagger}\hat{U}\hat{V}\ket{0}\right|^2\right] =
\frac{1}{N}, \label{eq:haar_average_2}
\end{equation}
where $N$ is the Hilbert space dimension.
Here, Eq.\eqref{eq:haar_average_2} holds because of the left- and right- unitary invariant property ($\mathbb{E}_{\hat{U},\hat{V}}[f(\hat{V}^{\dagger}\hat{U}\hat{V})]=\mathbb{E}_{\hat{U}}[f(\hat{U})]$) and a previous result \cite[Corollary 1.2]{https://doi.org/10.48550/arxiv.math/0608108}.
Furthermore, because $\mathbb{E}_{\hat{U}}[f(\hat{U})]=\mathbb{E}_{\hat{U}}[f(i\hat{U})]$ holds, we can show that
\begin{equation*}
\begin{split}
\mathbb{E}_{\hat{U}}\left[\mathrm{Re}[\braket{\psi|\hat{U}|\psi}]^2\right] &= \mathbb{E}_{\hat{U}}\left[\mathrm{Re}[\braket{\psi|i\hat{U}|\psi}]^2\right] \\
&=\mathbb{E}_{\hat{U}}\left[\mathrm{Im}[-\braket{\psi|\hat{U}|\psi}]^2\right] \\
&=\mathbb{E}_{\hat{U}}\left[\mathrm{Im}[\braket{\psi|\hat{U}|\psi}]^2\right].
\end{split}
\end{equation*}
Therefore, the shots in Eq.\eqref{eq:hadamard_shot_alloc} are evenly allocated ($m^{(\mathrm{r})}=m^{(\mathrm{i})}=m/2$), which results in the average variance,
\begin{equation}
\mathbb{E}_{\hat{U},\ket{\psi}}\left[\mathrm{Var}[Z]\right] =\frac{2}{m}\left(2 - \frac{1}{N}\right).
\end{equation}
Consequently, in the large Hilbert space ($N\to\infty$) and after averaging out the state and operator dependencies, the variance becomes as follow:
\begin{equation}
\mathrm{Var}[Z] \approx \frac{4}{m}, \label{eq:single_h_overlap_var}
\end{equation}
which results in Eq.\eqref{eq:var_mat_elem_S}.

\subsection{Hamiltonian Overlap}
\label{sec:hamiltonian_overlap}
Although obtaining a matrix element of $\bm{S}$ by the Hadamard test is straightforward, as written in Eq.\eqref{eq:single_h_overlap_var}, decomposition of the Hamiltonian by a unitary sum in Eq.\eqref{eq:pauli_ham_grp} should be used to evaluate the Hamiltonian overlap $\mathbb{E}[\bm{H}_{kl}] = \bra{\phi}\hat{V}_k^\dagger \hat{H} \hat{V}_l \ket{\phi}$, using a sequence of $N_\beta$ Hadamard tests, where $\hat{V}_k\ket{\phi} = \ket{\phi_k}$ and $\hat{V}_l\ket{\phi} = \ket{\phi_l}$, as below:
\begin{equation*}
\begin{split}
\bra{\phi}\hat{V}_k^\dagger \hat{H} \hat{V}_l \ket{\phi} =& \sum_{j=1}^{N_\beta} \beta_j \bra{\phi}\hat{V}_k^\dagger\hat{U}_j\hat{V}_l\ket{\phi} \\=:& \sum_{j=1}^{N_\beta} \beta_j \bra{\phi}\hat{U}'_j\ket{\phi}.
\end{split}
\end{equation*}

By using Eq.\eqref{eq:single_h_overlap_var} again, but with a distribution of given total shots $m=\sum_j^{N_\beta}m_j$, the merged sampling variance of the Hadamard tests is given as
\begin{equation}
\mathrm{Var}[\bm{H}_{kl}] \approx \sum_{j=1}^{N_\beta} \frac{4\beta_j^2}{m_j},
\end{equation}
which can be optimized with the following Lagrange multiplier approach:
\begin{equation}
\mathcal{L}(\{m_j\}_{j=1}^{N_\beta}; \lambda) = \sum_{j=1}^{N_\beta} \frac{4\beta_j^2}{m_j} + \lambda\left( \sum_{j=1}^{N_\beta}m_j - m \right),
\end{equation}
where $\lambda$ is a Lagrangian multiplier.
The resulting distribution of shots and optimized variance then become
\begin{align}
m_j &= \frac{\beta_j}{\sum_{j'} \beta_{j'}}m, \\
\mathrm{Var}[\bm{H}_{kl}]^\star &= \frac{4(\sum_j \beta_j)^2}{m} =: \frac{4\|\hat{H}\|^2_\beta}{m}, \label{eq:ham_overlap_var}
\end{align}
which result in Eq.\eqref{eq:var_mat_elem_H}.

From Eq.\eqref{eq:ham_overlap_var}, it is desirable to find the anticommutation grouping that minimizes $\|\hat{H}\|_\beta$.
As shown elsewhere \cite{efficient_quantum_measurement}, this formalism is similar to the measurement optimization in the VQE; however, finding such a grouping is NP-hard. Therefore, a heuristic algorithm was suggested instead. Furthermore, one can easily show that any grouping strategy is better than the individual Pauli grouping ($\hat{U}_j = \hat{P}_j$) because $\sum_j \beta_j=\sum_j (\sum_{k\in\mathcal{S}(j)}\alpha_k^2)^{1/2}\le\sum_k|\alpha_k|$.

\section{Behavior of the Random Matrices}
\label{sec:appendix_matrix_variance_statistics}
In this section, Theorems \ref{theorem:toeplitz_error_matrix} and \ref{theorem:non_toeplitz_error_matrix} are proved.
However, before that, the random matrix theory is introduced based on the matrix Laplace transform with slight modifications.

\begin{theorem}[Matrix Gaussian Series, {\cite[Theorem~4.1.1]{MAL-048}}]
	\label{theorem:matrix_gaussian_series}
	Let $\{\bm{A}_k\}$ be a fixed finite sequence of Hermitian matrices with dimensions $n$, and $\{{X}_k\}$ be an independent finite sequence of normal random variables with zero means and the variances of $\{\sigma^2_k\}$. Then, one can introduce a Gaussian matrix series,
	\begin{equation}
	\label{eq:sum_of_random_matrices}
	\bm{\Delta_Z}:=\sum_k {X}_k \bm{A}_k.
	\end{equation}
	Let $v(\bm{\Delta_Z})$ be the matrix variance statistic of the sum:
	\begin{equation}
	\label{eq:matrix_variance}
	v(\bm{\Delta_Z}) := \|\mathbb{E}[\bm{\Delta}_{\bm{Z}}^2]\| = \left\|\sum_k \sigma^2_k \bm{A}^2_k\right\|.
	\end{equation}
	Then, the expected norm of $\bm{\Delta_Z}$ is bounded as
	\begin{equation}
	\label{eq:expected_value_sum_of_random_matrices}
	\mathbb{E}[\|\bm{\Delta_Z}\|] \le \sqrt{2v(\bm{\Delta_Z})\log(2n)}.
	\end{equation}
	Furthermore, for all $t\ge 0$, the following inequality holds:
	\begin{equation}
	\label{eq:prob_sum_of_random_matrices}
	\mathbb{P}\{\|\bm{\Delta_Z}\|\ge t \} \le 2n \exp{\left(\frac{-t^2}{2v(\bm{\Delta_Z})}\right)}.
	\end{equation}
\end{theorem}

Theorem \ref{theorem:matrix_gaussian_series} is a modified version of its original version that adds the Hermitian condition and removes its unit variance condition.
This modification is made because, if a random variable $X$ follows the standard-normal distribution $\mathcal{N}(0, 1)$, its scaled random variable $\sigma X$ would have a distribution of $\mathcal{N}(0, \sigma^2)$ and vice versa.
Eq.\eqref{eq:expected_value_sum_of_random_matrices} provides a tight bound for the expected norm, whereas Eq.\eqref{eq:prob_sum_of_random_matrices} does not offer sufficient insight into the random matrix norm because its right-hand side is larger than 1 unless the deviation is considerably large $(t > \sqrt{2v(\bm{\Delta_Z})\log{2n}})$.
Hence, the use of the concentration inequality has been suggested \cite[Eq.4.1.8]{MAL-048}, treating $|\bm{\Delta_Z}|$ as a $v_{\star}$-Lipschitz function of the Gaussian random variables \cite[Theorem 5.6.]{10.1093/acprof:oso/9780199535255.002.0004}. 

\begin{theorem}[Matrix Gaussian Series Concentration, {\cite[Theorem~5.6.]{10.1093/acprof:oso/9780199535255.002.0004}}]
	\label{theorem:matrix_gaussian_series_concentration}
	The notation from Theorem \ref{theorem:matrix_gaussian_series} is adopted, and a weak variance of $\bm{\Delta_Z}$, denoted as $v_{\star}(\bm{\Delta_Z})$, is introduced, such that
	\begin{equation}
	\begin{split}
	v_{\star}(\bm{\Delta_Z})&:=\sup_{\|\bm{u}\|=\|\bm{v}\|=1}\mathbb{E}[|\bm{u}^{\dagger}\bm{\Delta_Z}\bm{v}|^2] \\
	&=\sup_{\|\bm{u}\|=\|\bm{v}\|=1}\sum_{k}\sigma_k^2|\bm{u}^{\dagger}\bm{A}_k\bm{v}|^2.
	\end{split}
	\end{equation}
	Then, the following concentration inequality holds because the weak variance is bounded as $v(\bm{\Delta_Z})/n\le v_{\star}(\bm{\Delta_Z}) \le v(\bm{\Delta_Z})$.
	\begin{equation}
	\label{eq:mat_norm_tail_bound}
	\mathbb{P}\left\{\|\bm{\Delta_Z}\| \ge \mathbb{E}[\|\bm{\Delta_Z}\|] + t\right\} \le 
	\exp{\left(\frac{-t^2}{2v_{\star}(\bm{\Delta_Z})}\right)}\le
	\exp{\left(\frac{-t^2}{2v(\bm{\Delta_Z})}\right)}.
	\end{equation}
\end{theorem}

Although this result alone fails to help elucidate $\mathbb{E}[\|\bm{\Delta_Z}\|]$, the tail bound centered by $\mathbb{E}[\|\bm{\Delta_Z}\|]$ is tighter than in Eq.\eqref{eq:prob_sum_of_random_matrices}.
Therefore, Theorem \ref{theorem:matrix_gaussian_series_concentration} is used instead of Eq.\eqref{eq:prob_sum_of_random_matrices} to represent the concentration of the random matrix.
Theorem \ref{theorem:matrix_gaussian_series_concentration} is also applied to derive the concentration inequality for $\|\bm{\Delta_Z}\|$ in Eq.\eqref{eq:mat_norm_concentration}.

In the following subsections, Theorems \ref{theorem:toeplitz_error_matrix} and \ref{theorem:non_toeplitz_error_matrix} are proved.
Although the efforts here are similar to previous work \cite[Chapter 4.2]{MAL-048}, two additional factors are introduced to meet the needs of the present application: the application to the complex random matrices induced by a sampling error and the optimal distribution of samplings that minimize the bound for matrix variances.
Specifically, individual complex error matrix elements describe the sampling errors of the Hadamard test with a finite number of shots and are mean zero and variance inversely proportional to the sampling number, as shown in Appendix \ref{sec:appendix_hadamard_sampling_variance}.
Furthermore, the total number of samplings is adequately distributed to minimize the bound for the expected norm (see Eq.\eqref{eq:expected_value_sum_of_random_matrices}).

\subsection{Toeplitz matrix}\label{sec:Toeplitz_mat_theory_proof}
Here, Theorem \ref{theorem:toeplitz_error_matrix} is proved.
\begin{proof}
Consider a Hermitian, Toeplitz and Gaussian random matrix $\bm{\Delta_Z}$ generated by a random sequence with $2n-1$ Gaussian random variables $\{{X}^{(\mathrm{r})}_{k}\}_{k=0}^{n-1}$ and $\{{X}^{(\mathrm{i})}_{k}\}_{k=1}^{n-1}$, where the sequence with a superscript $(\mathrm{r})$ denotes the real part of the matrix elements, and that with a superscript $(\mathrm{i})$ denotes the imaginary part.
The imaginary diagonal elements are omitted (${X}^{(\mathrm{i})}_0=0$) to satisfy the Hermitian condition.
Thus, $\bm{\Delta_Z}$ can be written as a Gaussian matrix series that is consistent with Eq.\eqref{eq:sum_of_random_matrices},
\begin{align}
\label{eq:sum_of_gaussian_matrices_toeplitz}
\bm{\Delta_Z} = {X}_{0}^{(\mathrm{r})}\bm{I}_n + \sum_{k=1}^{n-1}
\left[({X}_{k}^{(\mathrm{r})}+i{X}_{k}^{(\mathrm{i})})\bm{C}^k_1
+({X}_{k}^{(\mathrm{r})}-i{X}_{k}^{(\mathrm{i})})\bm{C}^k_{-1}\right],
\end{align}
where $\bm{I}_n$ is the $n\times n$ identity matrix, and $\bm{C}_1$ and $\bm{C}_{-1}$ are the superdiagonal and subdiagonal matrices shifted by $\pm 1$:
\begin{equation}
\bm{C}_1 := 
\begin{bmatrix}
0 & 1 & 0 & 0 & \cdots \\
0 & 0 & 1 & 0 & \cdots \\
0 & 0 & 0 & 1 & \cdots \\
0 & 0 & 0 & 0 & \cdots \\
\vdots & \vdots & \vdots & \vdots & \ddots
\end{bmatrix}, \quad
\bm{C}_{-1} := 
\begin{bmatrix}
0 & 0 & 0 & 0 & \cdots \\
1 & 0 & 0 & 0 & \cdots \\
0 & 1 & 0 & 0 & \cdots \\
0 & 0 & 1 & 0 & \cdots \\
\vdots & \vdots & \vdots & \vdots & \ddots
\end{bmatrix}.
\end{equation}
In addition, for a positive integer $k$, $\bm{C}^k_1$ and $\bm{C}^k_{-1}$ are offset by $\pm k$ instead of $\pm1$. For example, their squares are
\begin{equation}
\bm{C}^2_1 = 
\begin{bmatrix}
0 & 0 & 1 & 0 & \cdots \\
0 & 0 & 0 & 1 & \cdots \\
0 & 0 & 0 & 0 & \cdots \\
0 & 0 & 0 & 0 & \cdots \\
\vdots & \vdots & \vdots & \vdots & \ddots
\end{bmatrix}, \quad
\bm{C}^2_{-1} =
\begin{bmatrix}
0 & 0 & 0 & 0 & \cdots \\
0 & 0 & 0 & 0 & \cdots \\
1 & 0 & 0 & 0 & \cdots \\
0 & 1 & 0 & 0 & \cdots \\
\vdots & \vdots & \vdots & \vdots & \ddots
\end{bmatrix}.
\end{equation}
Therefore, Eq.\eqref{eq:sum_of_gaussian_matrices_toeplitz} describes a Hermitian Toeplitz matrix.

Furthermore, let us assume that the variances corresponding to each random Gaussian sequence are given as $\{\sigma_{k}^{(\mathrm{r})2}\}_{k=0}^{n-1}$ and $\{\sigma_{k}^{(\mathrm{i})2}\}_{k=1}^{n-1}$, respectively.
Then, the matrix variance $v(\bm{\Delta_Z})$ is given by
\begin{equation}
\label{eq:var_z_toeplitz_1}
\begin{split}
v(\bm{\Delta_Z}) =& \left\| \sigma_{0}^{(\mathrm{r})2}\bm{I}_n + \sum_{k=1}^{n-1}\sigma_{k}^{(\mathrm{r})2}
(\bm{C}_1^k+\bm{C}_{-1}^k)^2 - \sigma_{k}^{(\mathrm{i})2}(\bm{C}_1^k-\bm{C}_{-1}^k)^2 \right\|\\
=& \left\| \sigma_{0}^{(\mathrm{r})2}\bm{I}_n + \sum_{k=1}^{n-1}\sigma_{k}^{(\mathrm{r})2}(\bm{C}_1^{2k}+\bm{C}_{-1}^{2k}+\bm{D}_{k}+\bm{D}_{-k})- \sigma_{k}^{(\mathrm{i})2}(\bm{C}_1^{2k}+\bm{C}_{-1}^{2k}-\bm{D}_k-\bm{D}_{-k})\right\|,
\end{split}
\end{equation} 
where $\bm{D}_k$ and $\bm{D}_{-k}$ are diagonal matrices with elements of $1$ for the first and last $n-k$ diagonals, respectively, and created by the product of $\bm{C}_{1}^{k}$ and $\bm{C}_{-1}^{k}$.
As shown in Appendix \ref{sec:appendix_hadamard_sampling_variance}, the variances between the real and imaginary parts are made the same, $\sigma_{k}^{(\mathrm{r})2} = \sigma_{k}^{(\mathrm{i})2} =: \sigma_k^2/2$, canceling the $\bm{C}_{\pm 1}$ matrices.
Consequently, Eq.\eqref{eq:var_z_toeplitz_1} is rewritten as
\begin{equation}\label{eq:toeplitz_diag}
v(\bm{\Delta_Z}) = \left\| \sigma_0^2 \bm{I}_n + \sum_{k=1}^{n-1} \sigma_k^2(\bm{D}_k + \bm{D}_{-k})\right \| 
= \left\| \mathrm{diag}(\{v_l\}_{l=1\cdots n}) \right\|,     
\end{equation}
where
\begin{equation}\label{eq:toeplitz_minimax}
v_l := \sigma_0^2 + \sum_{k=1}^{l-1}2\sigma_k^2 + \sum_{k=l}^{n-l}\sigma_k^2.
\end{equation}
is defined for $l \le \lfloor n/2 \rfloor$ and $v_l:=v_{n-l}$ otherwise.

Applying the Hadamard test error model from Appendix \ref{sec:appendix_hadamard_sampling_variance}, $\sigma_k^2$'s are expressed in terms of the number of samplings $m_k$ demonstrated as Eq.\eqref{eq:single_h_overlap_var} or Eq.\eqref{eq:ham_overlap_var}, depending on $\bm{Z}$.
Subsequently, the variances in  Eq.\eqref{eq:toeplitz_minimax} are replaced by the numbers of shots approximated as positive real numbers, $\bm{m}=(m_0, m_1, \cdots m_{n-1})\in\mathbb{R}^{n}_{+}$, resulting in
\begin{equation}\label{eq:toeplitz_minimax_shot}
v_l(\bm{m}) = 2V_{\bm{Z}}^2\left(\frac{\delta_{\bm{Z}\bm{H}}}{m_0} + \sum_{k=1}^{l-1}\frac{4}{m_k} + \sum_{k=l}^{n-l}\frac{2}{m_k}\right),
\end{equation}
for $l\le \lfloor n/2 \rfloor$.
Here, the prefactor $V_{\bm{Z}}$ is determined as $V_{\bm{H}}=\|H\|_\beta$ and $V_{\bm{S}}=1$, which is analogous to Eqs.\eqref{eq:single_h_overlap_var} and \eqref{eq:ham_overlap_var}.
In addition, because the simplified measurement of the diagonal elements as shown in Eqs.\eqref{eq:diagonal_S} and \eqref{eq:diagonal_H}, $\sigma_0^2=2V_{\bm{Z}}^2\delta_{\bm{Z}\bm{H}}/m_0$ is used, where $\delta_{\bm{Z}\bm{H}}$ is one iff $\bm{Z}=\bm{H}$ and zero otherwise.

Moreover, each $v_l(\bm{m})$ for $l\le \lfloor n/2 \rfloor$ is a convex function and has a minimum at $\bm{m}^{(l)}$ with the constraint of the total number of shots ($\sum_k m_k = M_{\bm{Z}}$), which is given by
\begin{equation}\label{eq:minimax_opt_val}
v_l(\bm{m}^{(l)}) = \frac{2V_{\bm{Z}}^2}{M_{\bm{Z}}}\left(\sqrt{2}n-2(\sqrt{2}-1)l-2+\sqrt{2}+\delta_{\bm{Z}\bm{H}}\right)^2,
\end{equation}
\begin{equation}\label{eq:minimax_opt_shot}
\bm{m}^{(l)} = (\delta_{\bm{Z}\bm{H}}m_0^{(l)}, \underbrace{2m_0^{(l)}, \cdots, 2m_0^{(l)}}_{l-1}, \underbrace{\sqrt{2}m_0^{(l)}, \cdots, \sqrt{2}m_0^{(l)}}_{n-2l+1}, \underbrace{0, \cdots 0}_{l-1}),
\end{equation}
where $m_0^{(l)}:=M_{\bm{Z}}(\sqrt{2}n-2(\sqrt{2}-1)l-2+\sqrt{2}+\delta_{\bm{Z}\bm{H}})^{-1}$ satisfies the constraint.
Note that the optimal solution for $l>\lfloor n/2 \rfloor$ cases can be found with the above result because $v_{n-l}(\bm{m})=v_l(\bm{m})$.
Furthermore, for all $l\in \{2\cdots \lfloor {n/2}\rfloor\}$, $v_l(\bm{m})$ intersects with $v_1(\bm{m})$ at the point of $\bm{m}^{(1)}$,
\begin{equation}\label{eq:minimax_help}
v_l(\bm{m}^{(1)}) = v_1(\bm{m}^{(1)}),
\end{equation}
which can be checked by inserting Eq.\eqref{eq:minimax_opt_shot} with $l=1$ to Eq.\eqref{eq:toeplitz_minimax_shot}.

Because the spectral norm of the diagonal matrix in Eq.\eqref{eq:toeplitz_diag} is the maximum absolute value of the elements, the following expression holds:
\begin{equation}
v(\bm{\Delta_Z}(\bm{m})) = \max_{l\in\{1\cdots \lfloor {n/2} \rfloor\}}{v_l(\bm{m})}.
\end{equation}
The residual matrix, $\bm{\Delta_Z}$, is constructed with the distribution of shots of $\bm{m}$. 
Moreover, the minimization of $v(\bm{\Delta_Z})$ with respect to $\bm{m}$ leads to the minimum bound for the expected error matrix norm (see Eq.\eqref{eq:expected_value_sum_of_random_matrices}) and a minimax problem over the distribution of shots and the set of functions $v_l$,
\begin{equation}
v^{\mathrm{(opt)}}(\bm{\Delta_Z})=\min_{\bm{m}\in \mathcal{M}}\max_{l\in\mathcal{L}} v_l(\bm{m}),
\end{equation}
where $\mathcal{M}=\{\bm{m}\in\mathbb{R}^{+n}:\sum_k m_k = M_{\bm{Z}}\}$ and $\mathcal{L}=\{1,\cdots,\lfloor n/2 \rfloor\}$.
To solve the minimax problem, we define $w(\bm{m}):=\max_{l\in\mathcal{L}}v_l(\bm{m})$ whose lower bound is
\begin{equation}\label{eq:minimax_pf_0}
w(\bm{m}) \ge v_1(\bm{m}) \ge v_1(\bm{m}^{(1)}).
\end{equation}
We can also verify that 
\begin{align}
w(\bm{m}^{(1)})&=\max \{v_l(\bm{m}^{(1)}):l\in\mathcal{L}\} \label{eq:minimax_pf_1}\\
&=v_1(\bm{m}^{(1)})\label{eq:minimax_pf_2},
\end{align}
where Eq.\eqref{eq:minimax_pf_2} holds because Eq.\eqref{eq:minimax_help} is satisfied.
Finally, Eqs.\eqref{eq:minimax_pf_0} and \eqref{eq:minimax_pf_2} imply that
\begin{equation}
v^{\mathrm{(opt)}}(\bm{\Delta_Z})=v_1(\bm{m}^{(1)}).
\end{equation}
The corresponding optimal point is derived by substituting $l=1$ into Eqs.\eqref{eq:minimax_opt_val} and \eqref{eq:minimax_opt_shot}, resulting in
\begin{align}
m_{0} &= m_{0}^{(\mathrm{r})} = \frac{M_{\bm{Z}}\delta_{\bm{ZH}}}{\sqrt{2}(n-1)+1},\\
m_{k>0} &= m_{k>0}^{(\mathrm{r})} + m_{k>0}^{(\mathrm{i})} = \frac{M_{\bm{Z}}}{n-1+\delta_{\bm{ZH}}/\sqrt{2}},
\end{align}
\begin{align}
v^{\mathrm{(opt)}}(\bm{\Delta_Z}) &= \frac{2V_{\bm{Z}}^2}{M_{\bm{Z}}}((n-1)\sqrt{2}+\delta_{\bm{ZH}})^2\approx\frac{4V_{\bm{Z}}^2n^2}{M_{\bm{Z}}},
\end{align}
which is applied to Eq.\eqref{eq:expected_value_sum_of_random_matrices} to produce Eq.\eqref{eq:toeplitz_error_matrix_norm_bound}. 
\end{proof}

The results show that the off-diagonal terms are sampled $\sqrt{2}$ times more than the diagonal ones.
The observation includes an additional factor $n$ for the matrix variance statistic found in a previous report \cite[Chapter 4.4]{MAL-048}, which accounts for the increased sampling errors as $n$ increases, if the total number of samplings is limited.

\subsection{Non-Toeplitz matrix}\label{sec:Non_Toeplitz_mat_theory_proof}
In a manner similar to that in Appendix \ref{sec:Toeplitz_mat_theory_proof}, Theorem \ref{theorem:non_toeplitz_error_matrix} is proved.

\begin{proof}
For the non-Toeplitz Gaussian random matrix $\bm{\Delta_H}$ with a size of $n\times n$, imaginary parts are added to Theorem \ref{theorem:matrix_gaussian_series} by specifying the random Hermitian matrices as
\begin{equation}
\bm{\Delta_H} = \sum_{1\le k<l\le n} \left\{{X}_{k, l}^{(\mathrm{r})} (\bm{E}_{kl}+\bm{E}_{lk}) + i{X}_{k, l}^{(\mathrm{i})} (\bm{E}_{kl}-\bm{E}_{lk})\right\} + \sum_{1\le k \le n}{{X}_{k,k} \bm{E}_{kk}},
\end{equation}
where $\bm{E}_{kl}$ is a matrix with an element of $1$ at the $(k,l)$ index and $0$ elsewhere, and ${X}_{k,l}^{(\mathrm{r, i})}$ is the Gaussian random variable for the real and imaginary off-diagonal element parts and the real($k=l$) diagonal elements.
The total $n^2$ independent ${X}$'s are adopted construct the matrix. As shown in Appendix \ref{sec:appendix_hadamard_sampling_variance}, the variances are set to be equal to their complex counterparts, giving $\sigma^{(\mathrm{r})2}_{k,l}=\sigma^{(\mathrm{i})2}_{k,l}=\sigma^2_{k,l}/2$. Then, $v(\bm{\Delta_Z})$ is given as
\begin{align}
v(\bm{\Delta_H}) &= \left\|\sum_{1\le k<l\le n} \{\sigma^{(\mathrm{r})2}_{k,l}(\bm{E}_{kl}+\bm{E}_{lk})^2-\sigma^{(\mathrm{i})2}_{k,l}(\bm{E}_{kl}-\bm{E}_{lk})^2\} + \sum_{1\le k \le n} \sigma^2_{k,k}\bm{E}^2_{kk}\right\|\\
&=\left\| \sum_{1\le k < l \le n}\sigma^2_{k,l}(\bm{E}_{kk}+\bm{E}_{ll}) + \sum_{1\le k \le n}\sigma^2_{k,k}\bm{E}^2_{kk} \right\|\\
&=\left\|\mathrm{diag}\left\{\sigma_{k,k}^2+\sum_{1\le l < k}\sigma_{(l, k)}^2+\sum_{k < l \le n}\sigma_{(k, l)}^2\right\}_{1\le k \le n}\right\|\\
&=\max_{1\le k\le n}\left\{\sigma_{k,k}^2+\sum_{l\in [n]\setminus{\{k\}}}\sigma_{\mathrm{ord}(k,l)}^2\right\}  \label{eq:mat_var_non_toeplitz}
\end{align}
where $\mathrm{ord}(\cdot,\cdot)$ performs index ordering to keep the summation condition ($k<l$).
Subsequently, as we did in Appendix \ref{sec:hamiltonian_overlap}, $\sigma_{k,l}^2$ is replaced by an expression of samplings, $m_{k,l}$ (see Eq.\eqref{eq:single_h_overlap_var}):
\begin{align}
v(\bm{\Delta_H}(\bm{m})) = \max_{k\in [n]} 2\|\hat{H}\|_\beta^2\left( \frac{1}{m_{(k,k)}} + \sum_{l\in [n]\setminus\{k\}} \frac{2}{m_{\mathrm{ord}(k,l)}} \right)
=: \max_{k\in [n]} v_k(\bm{m}),
\end{align}
where the vector 
\begin{equation*}
\bm{m}=(m_{(1,1)}, m_{(1,2)}, \cdots m_{(n,n)}) \in \mathbb{R}_{+}^{n(n+1)/2}     
\end{equation*}
represents the distribution of shots to obtain the matrix element ${X}_{k,l}$.
Note that different notations are used for the function $v_k(\bm{m})$ and vector $\bm{m}$ compared to those in Appendix \ref{sec:Toeplitz_mat_theory_proof}.
	
The minimization of $v(\bm{\Delta_H})$ under $M_{\bm{H}}$ shots, leads to the following minimax problem:
\begin{equation}
v^{\mathrm{(opt)}}(\bm{\Delta_H}) = \min_{\bm{m}\in \mathcal{M}} \max_{k\in [n]} v_k(\bm{m}),
\end{equation}
where $\mathcal{M}=\{\bm{m}\in \mathbb{R}_+^{n(n+1)/2} : \sum_k m_k = M_{\bm{H}}\}$.
The minimax problem is solved by introducing two functions:
\begin{align}
w_M(\bm{m}) :=& \max_{k\in [n]} v_k(\bm{m}),\\
w_A(\bm{m}) :=& \frac{1}{n}\sum_{k\in[n]} v_k(\bm{m}).
\end{align}
Here, $w_A(\bm{m})$ is a convex function with respect to $\bm{m}$ and has a minimum point with the restriction of $\mathcal{M}$ at $\bm{m}^{(A)}$, whose elements are
\begin{equation}
m^{(A)}_{(k,k)}=\frac{M_{\bm{H}}}{n^2},
\end{equation}
and
\begin{equation}
m^{(A)}_{(k,l)}=\frac{2M_{\bm{H}}}{n^2},
\end{equation}
for all $k<l$. In addition, its minimum value is
\begin{equation}
w_A(\bm{m}^{(A)}) = \frac{2\|\hat{H}\|_\beta^2n^3}{M_{\bm{H}}}.
\end{equation}
Therefore, the following can be stated for all $\bm{m}\in \mathcal{M}$:
\begin{equation}
w_M(\bm{m}) \ge w_A(\bm{m}) \ge w_A(\bm{m}^{(A)}).
\end{equation}
Moreover, inserting $\bm{m}^{(A)}$ into $w_M(\bm{m})$ leads to
\begin{equation}
w_M(\bm{m}^{(A)}) = \frac{2\|\hat{H}\|_\beta^2n^3}{M_{\bm{H}}} = w_A(\bm{m}^{(A)}),
\end{equation}
because $v_k(\bm{m}^{(A)})$ has the same value for every $k\in [n]$.
Therefore, the minimax problem is solved.
Thus, the resulting optimized number of shots and $v(\bm{\Delta_H})$ are
\begin{align}
m_{(k,l)} &= \begin{cases}
\frac{2M_{\bm{H}}}{n^2} & k \neq l\\
\frac{M_{\bm{H}}}{n^2} & k = l
\end{cases},\\
v^{\mathrm{(opt)}}(\bm{\Delta_H}) &= \frac{2\|\hat{H}\|_\beta^2n^3}{M_{\bm{H}}}.
\end{align}
\end{proof}
Here, prefactor $\|\hat{H}\|_\beta$ is determined from Appendix \ref{sec:appendix_hadamard_sampling_variance}.
Furthermore, unlike in the Toeplitz case, the result shows that the off-diagonal terms should be sampled two times more than the diagonal terms.
Although, the variance statistics of a Hermitian standard-normal matrix has been argued to grow with order $n$ \cite[Chapter 4.2.1]{MAL-048}, we discovered that its scale increases to $n^3$ if the sampling complexity is considered (i.e., if the total number of samplings is limited). 

\section{A Simpler GEVP Perturbation}\label{sec:Simpler_GEVP_Perturbation}
In this section, we present another GEVP perturbation bound, which is related to the condition number in the linear equation problem, $\mathrm{cond}(\bm{S}):= \|\bm{S}\|\|\bm{S}^{-1}\|$.
\begin{theorem}[Relative GEVP perturbation based on Weyl's inequality]
	Assume that a GEVP of a matrix pair of a Hermitian $\bm{H}$ and positive definite $\bm{S}$ is solved by an eigenvalue of $E$ and the corresponding eigenvector $\bm{c}$.
	Also, consider the perturbed problem, $(\tilde{\bm{H}}, \tilde{\bm{S}}):=(\bm{H}+\bm{\Delta_H}, \bm{S}+\bm{\Delta_S})$, where $\tilde{\bm{H}}$ is Hermitian and $\tilde{\bm{S}}$ is positive definite.
	If small perturbation of 
	\begin{equation}\label{eq:small_S_pert_assumption}
	\|\bm{S}^{-1}\bm{\Delta_S}\| < 1         
	\end{equation}
	is assumed, the error in the perturbed eigenvalue is bounded as
	\begin{equation}\label{eq:th9_result}
	|\tilde{E} - E| \le \frac{\|\bm{H}\| \|\bm{S}^{-1}\|}{1-\|\bm{S}^{-1}\bm{\Delta_S}\|}\left( \mathrm{cond}(\bm{S})\frac{\|\bm{\Delta_S}\|}{\|\bm{S}\|} + \frac{\|\bm{\Delta_H}\|}{\|\bm{H}\|} \right).
	\end{equation}
\end{theorem}
\begin{proof}
	Since $\bm{S}$ is positive definite and invertible, the GEVP turns into a standard eigenvalue problem, $\bm{S}^{-1}\bm{H}\bm{c} = E^{(n)}\bm{c}$.
	Then, Weyl's inequality leads to an upper bound for the eigenvalue error as shown below:
	\begin{equation}\label{eq:pert_bound_th9}
	\begin{split}
	|\tilde{E}-E|&\le \| \tilde{\bm{S}}^{-1}\tilde{\bm{H}} - \bm{S}^{-1}\bm{H} \| \\
	&\le \| \tilde{\bm{S}}^{-1} - \bm{S}^{-1}\| \|\bm{H}\| + \|\tilde{\bm{S}}^{-1}\|\|\bm{\Delta_H}\|.         
	\end{split}
	\end{equation}
	Here, the $\tilde{\bm{S}}^{-1}=\bm{S}^{-1}(\mathbb{I}+\bm{S}^{-1}\bm{\Delta_S})^{-1}$ is expanded as
	\begin{equation}
	(\bm{S}+\bm{\Delta_S})^{-1} = \bm{S}^{-1}\sum_{k=0}^{\infty}(-\bm{S}^{-1}\bm{\Delta_S})^k,
	\end{equation}
	where the norm is bounded as
	\begin{equation}\label{eq:norm_bound_th9_1}
	\|(\bm{S}+\bm{\Delta_S})^{-1}\| \le \|\bm{S}^{-1}\|\sum_{k=0}^\infty \|\bm{S}^{-1}\bm{\Delta_S}\|^k 
	=\frac{\|\bm{S}^{-1}\|}{1-\|\bm{S}^{-1}\bm{\Delta_S}\|},
	\end{equation}
	with the condition of Eq.\eqref{eq:small_S_pert_assumption}.
	Likewise, $(\bm{S}+\bm{\Delta_S})^{-1} - \bm{S}^{-1}$ is expanded as
	\begin{equation}
	(\bm{S}+\bm{\Delta_S})^{-1} - \bm{S}^{-1} = \bm{S}^{-1}\sum_{k=1}^{\infty}(-\bm{S}^{-1}\bm{\Delta_S})^k.
	\end{equation}
	Correspondingly, we can show the norm bound
	\begin{equation}\label{eq:norm_bound_th9_2}
	\|(\bm{S}+\bm{\Delta_S})^{-1} - \bm{S}^{-1}\| \le \frac{\|\bm{S}^{-1}\|^2\|\bm{\Delta_S}\|}{1-\|\bm{S}^{-1}\bm{\Delta_S}\|},
	\end{equation}
	Finally, plugging Eqs.\eqref{eq:norm_bound_th9_1} and \eqref{eq:norm_bound_th9_2} into Eq.\eqref{eq:pert_bound_th9} results in Eq.\eqref{eq:th9_result}.
\end{proof}

\section{Consideration of Simulation Error}
\label{sec:appendix_consideration_of_simulation_error}
In this section, we discuss how the simulation error that results from the approximated simulation of $e^{-i\hat{H}t}$ affects the Toeplitz construction of the matrix $\bm{H}$.
Moreover, we propose an asymptotic criterion for circuit depth to determine a matrix construction method for $\bm{H}$ that incurs fewer total errors.
As reviewed in Section \ref{sec:QKSD}, the Toeplitz construction suffers from simulation errors and sampling errors, whereas the non-Toeplitz construction only suffers from sampling errors, that are larger than those of the Toeplitz construction.
Therefore, sufficient suppression of the simulation error in the Toeplitz construction leads to a smaller overall error than in the non-Toeplitz construction.
However, because it requires a longer circuit depth, the aim here is to determine the minimum circuit depth that makes the Toeplitz construction more advantageous.

We consider a Trotterization unitary, $\hat{U}_{\mathrm{ST}}(t)$, with a time step of $\Delta_{t}^{(\mathrm{ST})}$ that simulates the time propagation of the Hamiltonian, $e^{-i\hat{H}t}$.
The simulation error is defined as $\hat{\mathcal{E}}_{\mathrm{ST}}(t) = e^{-i\hat{H}t} - \hat{U}_{\mathrm{ST}}(t)$, whose magnitude is given as $\|\hat{\mathcal{E}}_{\mathrm{ST}}(t)\|=O(t^2/N_{\mathrm{ST}})$.
Here, we denote $N_{\mathrm{ST}}:=t/\Delta_t^{(\mathrm{ST})}$ as the number of Trotter steps.
In addition, the circuit realization of $\hat{U}_{\mathrm{ST}}(t)$ has a depth of $D=O(N_{\Gamma} N_{\mathrm{ST}})$, where $N_\Gamma$ is the number of Hamiltonian fragments which can be easily diagonalized.
Hence, the magnitude of the simulation error can be rewritten in terms of the circuit depth as
\begin{equation}
\|\hat{\mathcal{E}}_{\mathrm{ST}}(t)\|=O\left(\frac{N_\Gamma t^2}{D}\right).   
\end{equation}

Here, the approximated basis $\{\hat{U}_{\mathrm{ST}}(k\Delta_t)\ket{\phi_0}\}_{k=0}^{n-1}$ is used instead of the exact Krylov basis; thus, the ideal element of $\bm{H}$ becomes
\begin{equation}
[\bm{H}]_{kl} = \bra{\phi_0} \hat{U}_{\mathrm{ST}}^\dagger(k\Delta_t)\hat{H}\hat{U}_{\mathrm{ST}}(l\Delta_t) \ket{\phi_0}.
\end{equation}
The Toeplitz construction in Eq.\eqref{eq:prj_H_QKD} is possible with the help of the commutation relation, $[\hat{H}, e^{-i\hat{H}}]=0$.
However, because $\hat{U}_{\mathrm{ST}}$ does not commute with $\hat{H}$, the Toeplitz treatment may cause an error when obtaining the matrix element of $\bm{H}$.
The Toeplitz treatment of the element with the sampling noise leads to
\begin{align}
[\tilde{\bm{H}}^{(\mathrm{t})}]_{kl}&=\bra{\phi_0} \hat{H}\hat{U}_{\mathrm{ST}}((l-k)\Delta_t) \ket{\phi_0} + [\bm{\Delta}^{(\mathrm{t})}_{\bm{H},\mathrm{S}}]_{kl}\\
&=[\bm{H}]_{kl}
+ \bra{\phi_0} [\hat{H},\hat{U}_{\mathrm{ST}}^\dagger(k\Delta_t)]\hat{U}_{\mathrm{ST}}(l\Delta_t) \ket{\phi_0} + [\bm{\Delta}^{(\mathrm{t})}_{\bm{H},\mathrm{S}}]_{kl}\\
&=[\bm{H}]_{kl}
+ \bra{\phi_0} [\hat{H},\hat{\mathcal{E}}_{\mathrm{ST}}^\dagger(k\Delta_t)]\hat{U}_{\mathrm{ST}}(l\Delta_t) \ket{\phi_0} + [\bm{\Delta}^{(\mathrm{t})}_{\bm{H},\mathrm{S}}]_{kl}\\
&=[\bm{H}]_{kl} + [\bm{\Delta}_{\bm{H},\mathrm{C}}^{(\mathrm{t})}]_{kl} + [\bm{\Delta}_{\bm{H},\mathrm{S}}^{(\mathrm{t})}]_{kl},
\end{align} 
 where $[\bm{\Delta}^{(\mathrm{t})}_{\bm{H},\mathrm{S}}]_{kl}=O(\|\hat{H}\|_\beta \sqrt{n/M_{\bm{H}}})$ is sampling error in the Toeplitz sequence (see Eq.\eqref{eq:QKD_H_seq}), whose magnitude is derived in Appendix \ref{sec:hamiltonian_overlap}.
In addition to the sampling error, a commutation error appears in the Toeplitz construction of $\bm{H}$, which is
\begin{align}
[\bm{\Delta}_{\bm{H},\mathrm{C}}^{(\mathrm{t})}]_{kl}:=&\bra{\phi_0} [\hat{H},\hat{\mathcal{E}}_{\mathrm{ST}}^\dagger(k\Delta_t)]\hat{U}_{\mathrm{ST}}(l\Delta_t) \ket{\phi_0}\\
=&O\left(\|\hat{\mathcal{E}}_{ST}(k\Delta_t)\| \|\hat{H}\| \right)\\
=&O\left( \frac{N_\Gamma \|\hat{H}\| \Delta_t^2}{D}k^2 \right).
\end{align}
A bound for spectral norm of $\bm{\Delta}_{\bm{H},\mathrm{C}}^{(\mathrm{t})}$ can be derived from the Frobenius norm ($\|\bm{\Delta}_{\bm{H},\mathrm{C}}^{(\mathrm{t})}\|\le\|\bm{\Delta}_{\bm{H},\mathrm{C}}^{(\mathrm{t})}\|_\mathrm{F}$):
\begin{equation}
\|\bm{\Delta}_{\bm{H},\mathrm{C}}^{(\mathrm{t})}\|_\mathrm{F} = \left(\sum_{kl}[\bm{\Delta}_{\bm{H},\mathrm{C}}^{(\mathrm{t})}]_{kl}^2\right)^{1/2}=O\left( \frac{N_\Gamma \|\hat{H}\| \Delta_t^2 n^3}{D} \right).
\end{equation}

Finally, triangular inequality is used to analyze the total error matrix norm,
\begin{align}\label{eq:toeplitz_tot_err_norm_comm}
\begin{split}
\|\bm{\Delta}^{(\mathrm{t})}_{\bm{H}}\|\le& \| \bm{\Delta}^{(\mathrm{t})}_{\bm{H},\mathrm{C}} \| + \| \bm{\Delta}^{(\mathrm{t})}_{\bm{H},\mathrm{S}} \| \\
=& O\left(\frac{N_\Gamma \|\hat{H}\| \Delta_t^2 n^3}{D} + \frac{\|\hat{H}\|_\beta n \sqrt{\log n}}{\sqrt{M_{\bm{H}}}}\right),     
\end{split}
\end{align}
where the spectral norm of $\bm{\Delta}^{(\mathrm{t})}_{\bm{H},\mathrm{S}}$ is obtained from Theorem \ref{theorem:toeplitz_error_matrix}.

Meanwhile, for the non-Toeplitz case, only the sampling error whose upper bound is revealed in Theorem \ref{theorem:non_toeplitz_error_matrix} is considered.
Although only the upper bound is shown, we assume that the result in Theorem \ref{theorem:non_toeplitz_error_matrix} also explains the asymptotic lower bound based on the observation shown in the numerical analysis (Fig.\ref{fig:error_norm}). In other words,
\begin{equation}\label{eq:nontoeplitz_tot_err_norm_comm}
\|\bm{\Delta}_{\bm{H}}^{(\mathrm{nt})}\| = \Theta\left( \frac{\|\hat{H}\|_\beta n\sqrt{n\log n}}{\sqrt{M_{\bm{H}}}} \right)
\end{equation}

To determine which construction is advantageous, one can compare Eqs.\eqref{eq:toeplitz_tot_err_norm_comm} and \eqref{eq:nontoeplitz_tot_err_norm_comm}.
Correspondingly, one can show an asymptotic condition for the quantum circuit depth that makes the Toeplitz construction more advantageous than the non-Toeplitz construction ($\|\bm{\Delta}_{\bm{H}}^{(\mathrm{t})}\| \le \|\bm{\Delta}_{\bm{H}}^{(\mathrm{nt})}\|$):
\begin{equation}\label{eq:circuit_depth_threshold}
D^{\star} = \Omega\left( N_{\Gamma}\Delta_t^2 \left(\frac{\|\hat{H}\|}{\|\hat{H}\|_\beta}\right) \left(\frac{n^3M_{\bm{H}}}{\log n}\right)^{1/2} \right).
\end{equation}
In other words, if a quantum computer offers circuits longer than Eq.\eqref{eq:circuit_depth_threshold}, the Toeplitz construction suffers less error than the non-Toeplitz construction.
The dependency of $M_{\bm{H}}$ explains that the simulation error should be suppressed as sampling error is reduced in order to make the Toeplitz construction advantageous.
Moreover, because the time step for the Krylov basis, $\Delta_t$, determines the maximum simulation time $(t=n\Delta_t)$, the circuit depth becomes proportional to $\Delta_t$.

The corresponding eigenvalue perturbation scales differently depending on whether the circuit depth of Eq.\eqref{eq:circuit_depth_threshold} is satisfied, due to the different construction strategies involved:
\begin{equation}\label{eq:pert_bound_st_1}
|\tilde{E}_{0,\mathrm{ST}}^{(n\rightarrow n_\epsilon)} - {E}_{0,\mathrm{ST}}^{(n\rightarrow n_\epsilon)}| \le
\left\{
\begin{array}{cc}
\tilde{O}\left(\frac{\|\hat{H}\|^3 n^2}{d_0\sqrt{M}}\right)&D > D^{\star}\\
\tilde{O}\left(\frac{\|\hat{H}\|^3 n^{5/2}}{d_0\sqrt{M}}\right)&D \le D^{\star}
\end{array}
\right..
\end{equation}
Here, $E_{0,\mathrm{ST}}^{(n\rightarrow n_\epsilon)}$ denotes the lowest eigenvalue in the approximated subspace, $\tilde{\mathcal{K}}_n:=\{\hat{U}_{\mathrm{ST}}(k\Delta_t)\ket{\phi_0}\}_{k=0}^{n-1}$, and $\tilde{E}_{0,\mathrm{ST}}^{(n\rightarrow n_\epsilon)}$ is the eigenvalue perturbed by the sampling and Trotter errors, where the thresholding with the parameter $\epsilon$ is employed.

\section{Hardware Noise}
\label{sec:appendix_Hardware_noise}
Here, we partially review an analysis of Hardware noise in \cite{liang2023modeling}, which explains how the erroneous quantum operations contribute to the final error in the Randomized Fourier Estimation (RFE) algorithm as a case study in the early fault tolerant regime.
Furthermore, we demonstrate how this noise model is applied to the QKSD algorithm.

A noisy circuit is modeled by placing independently selected random circuits between the ideal unitaries.
Such setting results in the following mixed state $\rho_f$, evolved from the pure initial state, $\ket{\Psi_{i}}\bra{\Psi_{i}}$:
\begin{equation}\label{eq:rho_f}
\rho_f = \Lambda \circ \mathcal{W}_D \circ \cdots \circ \Lambda \circ \mathcal{W}_1(\ket{\Psi_{i}}\bra{\Psi_{i}}),
\end{equation}
where $\mathcal{W}_k(\rho) := \hat{W}_k \rho \hat{W}_k^{\dagger}$ is a superoperator describing the action of an ideal unitary, $\hat{W}_k$, which constitutes the entire circuit with the depth of $D$.
Also, the error channel is described by the following Kraus decomposition with the complete set of $N_q$-qubit Pauli operators ($\hat{P}_j \in \{ \hat{I}, \hat{\sigma}_x, \hat{\sigma}_y, \hat{\sigma}_z\}^{N_q}$):
\begin{equation}
\Lambda(\rho) = r_0\rho + \sum_{j=1}^{4^{N_q}-1}r_j\hat{P}_j \rho \hat{P}_j^\dagger,
\end{equation}
where $r_j\in [0, 1]$ is the probability that $\hat{P}_j$ occurs.
Note that Eq.\eqref{eq:rho_f} represents the logical qubit state, and thus, $(1-r_0)$ is analogous to the logical qubit error rate.
Then, the final mixed state (Eq.\eqref{eq:rho_f}) is represented as
\begin{equation}
\rho_f = f_{0} \ket{\Psi_{0}}\bra{\Psi_{0}} + \sum_{j=1}^{4^{N_q D} - 1}f_j \ket{\Psi_j}\bra{\Psi_j},
\end{equation}
where the ideal state $\ket{\Psi_0}:=\hat{W}_D \cdots \hat{W}_1 \ket{\Psi_i}$ is found with the probability of $f_0$, while the erroneous trajectory states, $\ket{\Psi_{j}}$ for $j>0$, appear with the probability of $f_j$.
For a simpler analysis, the erroneous trajectory states are assumed to be sampled from the $N$-dimensional complex spherical 2-design.
Under this assumption, the corresponding sampling results in the following expectation values and variance for the observable $\hat{O}$,
\begin{align}
\mathbb{E}_{\Psi_{j>0}\sim \mathcal{S}^{N}}[\mathrm{Tr}[\rho_f \hat{O}]] =& f_0 \braket{\Psi_0|\hat{O}|\Psi_0}\label{eq:noisy_expectation}\\
\mathrm{Var}_{\Psi_{j>0}\sim \mathcal{S}^{N}}[\mathrm{Tr}[\rho_f \hat{O}]] =& \frac{1}{N+1} \sum_{j=1}^{4^{N_qD}-1}f_j^2.\label{eq:noisy_variance}
\end{align}
If the error rates are much smaller than $f_0$, the fluctuation to $\mathrm{Tr}[\rho_f \hat{O}]$ (square root of Eq.\eqref{eq:noisy_variance}) becomes negligible compare to Eq.\eqref{eq:noisy_expectation}.
Therefore, $\mathrm{Tr}[\rho_f \hat{O}]$ will not largely deviate from Eq.\eqref{eq:noisy_expectation}.
With a further assumption of Pauli errors which are identical and independent to the qubits, it can be shown that $f_0=r_0^{D}=r^{N_q D}$, where $r$ is single qubit fidelity.

Then, the probabilities of the outcome in our Hadamard test settings in Eq.\eqref{eq:probability_hadamard} are respectively modified to
\begin{align}
\nonumber p_{\mathrm{noisy}} =& \frac{1}{2}\left( 1 + e^{-\lambda}\mathbb{E}[R] \right) \\ 
q_{\mathrm{noisy}} =& \frac{1}{2}\left( 1 + e^{-\lambda}\mathbb{E}[I] \right),
\end{align}
where $e^{-\lambda} = r^{N_q D}$.
Therefore, the errors change the expectation value (Eq.\eqref{eq:hadamard_noiseless_expectation}) and the variance (Eq.\eqref{eq:hadamard_noiseless_variance}) of the Hadamard test outcome as
\begin{align}
\mathbb{E}[Z_{\mathrm{noisy}}] =& e^{-\lambda}\braket{\psi|\hat{U}|\psi}\\
\mathrm{Var}[Z_{\mathrm{noisy}}] =& \frac{2}{m}\left(2 - e^{-2\lambda}\left|\braket{\psi|\hat{U}|\psi}\right|^2\right). \label{eq:hardware_variance}
\end{align}
However, the approximation to variance by removing the state and operator dependencies produces the same result as Eq.\eqref{eq:single_h_overlap_var}.
Therefore, the Pauli error results in an exponentially vanishing value with respect to the circuit depth and the number of qubits, while the other errors remain nearly the same.

Let us assume the identical circuit depth is used for each Hadamard test, and thus $\lambda$ is invariant over the matrix elements.
Consequently, the QKSD matrix decays with the constant factor
\begin{equation}\label{eq:hardware_noise_matrix}
\tilde{\bm{Z}}_{\mathrm{H}}=e^{-\lambda}\bm{Z} + \bm{\Delta}_{\bm{Z},\mathrm{H}}\approx e^{-\lambda}\bm{Z} + \bm{\Delta}_{\bm{Z}},
\end{equation}
where $\tilde{\bm{Z}}_{\mathrm{H}}$ is a QKSD matrix ($\bm{Z}\in\{\bm{H}, \bm{S}\}$) with the hardware noise and $\bm{\Delta}_{\bm{Z},\mathrm{H}}$ denotes the statistical error with the uncertainty of Eq.\eqref{eq:hardware_variance}, which behaves almost identical to $\bm{\Delta_Z}$.
When $\bm{\Delta}_{\bm{Z},\mathrm{H}}$ is small ($e^{\lambda} \|\bm{\Delta}_{\bm{Z}}\| \ll 1$), the multiplicative error is canceled, since the GEVP(Eq.\eqref{eq:gen_eigeq}) remains invariant ($e^{-\lambda}\bm{H}\bm{c}_j=e^{-\lambda}\bm{S}\bm{c}_j E^{(n)}_j$).
However, in general, $\bm{\Delta_Z}$ is amplified by the factor of $e^{\lambda}$, because replacing the matrices in Eq.\eqref{eq:gen_eigeq} with Eq.\eqref{eq:hardware_noise_matrix} results in Eq.\eqref{eq:hardware_noise_gevp}

\end{document}